\documentclass[
final
]{dmtcs-episciences}

\usepackage[utf8]{inputenc}
\usepackage{subfigure}

\usepackage{amssymb,amstext}
\usepackage{cite,scrtime}
\usepackage{graphics}
     \graphicspath{{fig/}}
\usepackage{amsmath,dsfont}
\usepackage{mathrsfs}
\usepackage{fancyhdr}
\usepackage{verbatim}
\usepackage{boxedminipage}
\usepackage{colortbl}
\usepackage[vlined]{algorithm2e}
\usepackage[pdftex]{hyperref}

\usepackage{graphicx}

\renewcommand{\leq}{\leqslant}
\renewcommand{\geq}{\geqslant}

\def\debut{\begin{itemize}\item[{\bf [[}]\small}
\def\term{\hfill {\bf]]} \end{itemize} }

\newcommand{\N}{\hbox{I\hskip -2pt N}}

\newcommand{\G}{\mathcal{G}}

\newtheorem{theorem}{Theorem}
\newtheorem{proposition}{Proposition}
\newtheorem{lemma}{Lemma}
\newtheorem{definition}{Definition}

\newenvironment{proof}[1][]{\par \noindent {\bf Proof:#1}\ }{\hfill$\Box$\\}
\newenvironment{proofFINAL}[1][]{\par \noindent {\bf Proof of Theorem~\ref{th:main}:#1}\ }{\hfill$\Box$\\}

\newcommand{\YES}{\textsc{Yes}}
\newcommand{\NO}{\textsc{No}}
\newcommand{\tds}   {TDS\xspace}

\usepackage{xspace}
\newcommand{\name}[1]{\textsc{#1}}

\newcommand{\dom}	{\name{Dominating Set}\xspace}
\newcommand{\totdom}{\name{Total Dominating Set}\xspace}
\renewcommand{\G}		{\ensuremath{G=(V,E)}\xspace}
\newcommand{\GB}	{\ensuremath{G'=(V',E')}\xspace}

\renewcommand{\N}[1]	{\ensuremath{N_{#1}(v,w)}\xspace}

\newcommand{\vw}	{\ensuremath{\{ v,w\}}\xspace}
\graphicspath{{./}}

\newcommand{\mc}{\mathcal}

\newtheorem{fact}{Fact}
\newtheorem{rgl}  {Rule}

\newcommand{\Dvw}   {\ensuremath{\mc{D}}\xspace}
\newcommand{\Dv}    {\ensuremath{\mc{D}_{v}}\xspace}
\newcommand{\Dw}    {\ensuremath{\mc{D}_{w}}\xspace}

\newcommand{\add}   [1] {\textcolor{red} {#1}}

\newcommand{\modif} [2] {\add{#2}}

\newcommand{\modifOK} [2] {#2}
\newcommand{\modifNO} [2] {}

\author{Valentin Garnero\affiliationmark{1}
  \and Ignasi Sau\affiliationmark{1}\thanks{Supported by projects DEMOGRAPH (ANR-16-CE40-0028) and ESIGMA (ANR-17-CE40-0028).}}
\title[A Linear Kernel for Planar Total Dominating Set]{A Linear Kernel for Planar Total Dominating Set}
\affiliation{
  CNRS, AlGCo project-team, LIRMM, Universit\'e de Montpellier, Montpellier, France}
\keywords{parameterized complexity, planar graphs, linear kernels, total domination}

\received{2017-5-2}
\revised{2018-12-13}
\accepted{2018-4-19}

\begin{document}
\publicationdetails{20}{2018}{1}{14}{3295}
\maketitle
\begin{abstract}
 A \emph{total dominating set} of a graph $G=(V,E)$ is a subset $D \subseteq V$ such that every vertex in $V$ is adjacent to some vertex in $D$. Finding a total dominating set of minimum size is NP-hard on planar graphs and $W[2]$-complete on general graphs when parameterized by the solution size. By the meta-theorem of Bodlaender et al.~[J. ACM, 2016],  there {\sl exists} a linear kernel for \textsc{Total Dominating Set} on graphs of bounded genus.
Nevertheless, it is not clear how such a kernel can be effectively {\sl constructed}, and how to obtain {\sl explicit} reduction rules with reasonably small constants.
Following the approach of Alber et al.~[J. ACM, 2004], we provide an explicit  kernel for \textsc{Total Dominating Set} on planar graphs with at most $410k$ vertices, where $k$ is the size of the solution.  This result complements several known constructive linear kernels on planar graphs for other domination problems such as \textsc{Dominating Set}, \textsc{Edge Dominating Set}, \textsc{Efficient Dominating Set}, \textsc{Connected Dominating Set}, or \textsc{Red-Blue Dominating Set}. 
\end{abstract}

\section{Introduction}
\label{sec:intro}


The field of parameterized complexity deals with algorithms for decision problems whose instances consist of a pair $(x,k)$, where~$k$ is a secondary measurement known as the \emph{parameter}.
A fundamental concept in this area is that of \emph{kernelization}. A kernelization
algorithm\modif{, or just \emph{kernel}}{} for a parameterized problem takes an
instance~$(x,k)$ of the problem and, in time polynomial in $|x| + k$, outputs
an equivalent instance~$(x',k')$, \modifOK{}{called \emph{kernel}}, such that $\max\{|x'|, k'\} \leq g(k)$ for some
function~$g$. The function~$g$ is called the \emph{size} of the kernel and may
be viewed as a measure of the ``compressibility'' of a problem using
polynomial-time preprocessing rules.
In order to ensure an efficient compression, it is natural to look for polynomial or linear kernels, that is, kernels with polynomial or linear size, respectively. (In the rest of the article we will not  formally distinguish between the two notions of kernel and kernelization.)
For an introduction to parameterized complexity and kernelization see the main textbooks on this field~\cite{DF13,FG06,Nie06,CyganFKLMPPS15}.



During the last decade, a plethora of results emerged on linear
kernels for graph-theoretic problems restricted to {\sl sparse} graph classes, that is, classes of graphs for which the number of edges depends linearly on the number of vertices. A pioneering result in this area is the linear
kernel for \textsc{Dominating Set} on planar graphs by Alber et al.~\cite{AFN04}, which gave rise to an explosion of results on
linear kernels on planar graphs and other sparse graph classes. Let us just mention some of the most important ones. Following the ideas of Alber et al.~\cite{AFN04}, Guo and Niedermeier~\cite{GuNi07} designed a general
framework and showed that problems that satisfy a certain \emph{distance property}
admit linear kernels on planar graphs. This result was subsumed by that of
Bodlaender et al.~\cite{BFL+09} who provided a meta-theorem for
problems to have a linear kernel on graphs of bounded genus. Later Fomin et al.~\cite{FLST10} extended
these results for \emph{bidimensional} problems to an even larger graph class, namely,
$H$-minor-free and apex-minor-free graphs. The most general result in this area is by Kim et al.~\cite{KLP+12}, who provided linear kernels for \emph{treewidth-bounding} problems on $H$-topological-minor-free graphs. (Note that in all these works, the problems are parameterized by the solution size.)

A common feature of these meta-theorems on sparse graphs is a {\sl
decomposition scheme} of the input graph that, loosely speaking, allows to deal
with each part of the decomposition independently. For instance, the approach
of Guo and Niedermeier~\cite{GuNi07}, which strongly builds on Alber et al.~\cite{AFN04},  is to consider a
so-called \emph{region decomposition} of the input planar graph. The key point
is that in an appropriately reduced \YES-instance, there are $O(k)$ regions and
each one has constant size, yielding the desired linear kernel. This idea was
generalized by Bodlaender et al.~\cite{BFL+09} to graphs on surfaces, where the role of regions
is played by \emph{protrusions}, which are graphs with small treewidth and
small boundary.  A crucial point is that while the reduction rules of~\cite{AFN04,GuNi07} are {\sl \modifOK{problem-dependent}{designed for a {\sl fixed} problem or a fixed set of problems}}, those of \modifOK{}{Bodlander et al.}~\cite{BFL+09} are {\sl automated}\footnote{\modifOK{}{Strictly speaking, the rules of Bodlander et al.~\cite{BFL+09} may also be considered as problem-dependent, but they are automatically generated by a unique ``meta-reduction rule'', which is general enough to cover many problems.}}, relying on a property called \emph{finite integer index} (FII), which was introduced by Bodlaender and de
Fluiter~\cite{BvF01}. Informally speaking, having FII guarantees that ``large'' protrusions of a graph can be replaced with ``small'' gadget graphs preserving equivalence of instances. FII\modifOK{}{ (together with additional properties such as the  so-called {\sl separation property}),} is also of central importance to the approach of Fomin et al.~\cite{FLST10} (resp. Kim et al.~\cite{KLP+12}) on $H$-minor-free (resp. $H$-minor-topological-free) graphs. See~\cite{BFL+09,FLST10,KLP+12} for more details.


Although of great theoretical importance, the aforementioned meta-theorems have two important drawbacks from a practical point of view. On the one hand, and more importantly, these results relying on FII guarantee the {\sl existence} of a linear kernel, but nowadays it is still not clear how such a kernel can be effectively {\sl constructed}. On the other hand, even if we knew how to construct such kernels, at the price of generality one cannot hope them to provide explicit reduction rules and small constants for a particular graph problem. Summarizing, as mentioned explicitly by Bodlaender et al.~\cite{BFL+09}, these meta-theorems provide simple criteria to decide whether a problem admits a linear kernel on a graph class, but finding linear kernels with reasonably small constant factors for concrete problems remains a worthy  topic for investigation. It is worth mentioning that Garnero et al.~\cite{GPST14} presented some first advances in this direction of obtaining explicit linear  kernels on sparse graphs, but their approach covers only some particular families of problems, and the involved constants are still quite large.

In this article we follow this research avenue and focus on the \textsc{Total Dominating Set} problem on planar graphs. A \emph{total dominating set} (or \tds for short) of a graph~$G=(V,E)$ is a subset $D \subseteq V$ such that every vertex in $V$ is adjacent to some vertex in $D$ (equivalently, a total dominating set is a dominating set inducing a subgraph without isolated vertices). In the decision version of the problem, we are given a graph $G$ and an integer $k$, and the objective is to find a \tds of size at most $k$ in $G$, or correctly report that such a set does not exist.  Total domination was introduced by Cockayne, Dawes, and Hedetniemi~\cite{CDH80} almost four decades ago and has remained a very active topic in graph theory since then (cf.~\cite{HJS12,HeYe13,HeYe12,ThYe07} for a few examples). More details and references can be found in the comprehensive survey of Henning~\cite{Hen09} or the book of Haynes, Hedetniemi, and Slater~\cite{HHS98}.
Fig.~\ref{img_dom} gives an example showing that, in particular, a minimal dominating set is not necessarily a subset of a minimal \tds, or vice versa.

\begin{figure}[h]
\begin{center}
\begin{tabular}{cc}
   \includegraphics[scale=0.7]{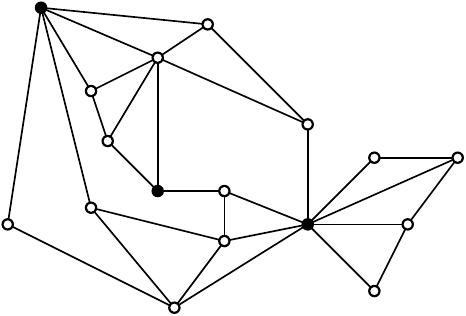}  $\ \ $  & $\  \ $
   \includegraphics[scale=0.7]{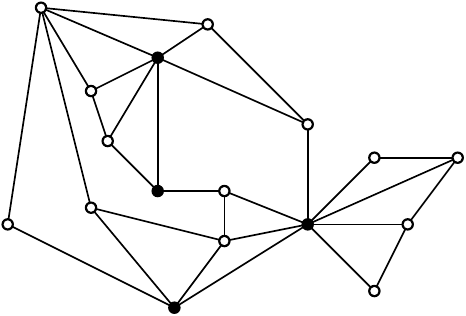} \\
   (a) & (b)
\end{tabular}
\end{center}
\vspace{-.35cm}
   \caption{The vertices depicted with $\bullet$ define a minimal (a) dominating set; (b) \tds.}
   \vspace{-.2cm}
   \label{img_dom}
\end{figure}


From a (classical) complexity point of view, finding a \tds of minimum size is {\sf NP}-hard on planar graphs~\cite{GJ79,Zhu09}\footnote{This result is indeed claimed in the literature~\cite{GJ79,Zhu09}, but we were not able to find any proof of it.}. From a parameterized complexity perspective (see~\cite{DF13,FG06,Nie06,CyganFKLMPPS15} for the missing definitions), \textsc{Total Dominating Set} parameterized by the size of the solution is $W[2]$-complete on general graphs~\cite{DF13} and {\sf FPT} on planar graphs~\cite{FoTh06,ABF+02}.


%
%

\paragraph{\textbf{\emph{Our results and techniques}}.} In this article we provide the first explicit (and reasonably simple) polynomial-time data reduction rules for \textsc{Total Dominating Set} on planar graphs, which lead to a linear kernel for the problem. In particular, we prove the following theorem.

\begin{theorem}
\label{th:main}
The \textsc{Total Dominating Set} problem parameterized by the solution size admits a linear kernel on planar graphs. More precisely, there exists a polynomial-time algorithm that given a planar instance $(G,k)$, either correctly reports that  $(G,k)$ is a \NO-instance or returns an equivalent instance $(G',k)$ such that $|V(G')| \leq 410 \cdot k$.
\end{theorem}


This result complements several explicit linear kernels on planar graphs for other domination problems such as \textsc{Dominating Set}~\cite{AFN04}, \textsc{Edge Dominating Set}~\cite{GuNi07}, \textsc{Efficient Dominating Set}~\cite{GuNi07}, \textsc{Connected Dominating Set}~\cite{LMS11,GuIm10}, or \textsc{Red-Blue Dominating Set}~\cite{GST13}. Although it is arguable whether the constant in Theorem~\ref{th:main} is {\sl small} or not, it is worth mentioning that our constant is comparable to the constants given for \textsc{Dominating Set}~\cite{AFN04} (even if this constant has been subsequently improved\modif{ in}{}~\cite{CFKX07}), for \textsc{Connected Dominating Set}\modif{ in}{}~\cite{GuIm10}, or for \textsc{Maximum Triangle Packing}\modif{ in}{}~\cite{GuNi07}, respectively, and that a much larger constant has been given for \textsc{Connected Dominating Set}\modif{ in}{}~\cite{LMS11}. Let us also mention that, as discussed in Section~\ref{sec:reduction}, our kernelization algorithm runs in cubic time. We believe that both the constant and the running time can be improved, but in this article our main objective was  to provide a reasonably small explicit kernel while keeping, as far as possible, the technical details not too complicated.

Our techniques are much inspired from those of Alber~et al.~\cite{AFN04} for \textsc{Dominating Set}. Roughly speaking, the idea of the method is to consider the neighborhood of each vertex and the neighborhood of each pair of vertices, and to identify some vertices that can be removed without changing the size of a smallest total dominating set\modifNO{}{, we moreover consider a specific part of the neighborhood of pairs, called simple region, in which we can identify additional vertices that can be removed}. The corresponding reduction rules are called Rule~\ref{rgl: Tot seul} for a single vertex, and Rules~\ref{rgl: Tot paire} and~\ref{rgl: Tot aux} for a pair of vertices. Crucial to this approach is to decompose the planar input graph (which we assume to be already embedded in the plane; such a graph is called \emph{plane}) into so-called \emph{regions}, which contain all vertices but $O(k)$ of them. Then it just remains to prove that in a reduced plane graph (we say that a graph $G$ is \emph{reduced} under a set of rules if none of these rules can be applied to $G$ anymore, or if the graph remains unchanged after their application) the total number of regions is $O(k)$ and that each of them contains $O(1)$ vertices, implying that the total number of vertices in the reduced instance is $O(k)$.

The main difference of our approach with respect to Alber et al.~\cite{AFN04} lies in Rules~\ref{rgl: Tot paire} and~\ref{rgl: Tot aux}. More precisely, due to the particularities of our problem, we need to distinguish more possibilities according to the neighborhood of a pair of vertices, which makes our reduction rules slightly more involved. In particular, while in \modifOK{}{the kernel for \dom}~\cite{AFN04} the region decomposition is only used for the analysis of the reduced graph, we also use regions in order to reduce the input graph\modifNO{}{(with Rule~3)}. We would like to mention that regions are also used in some of the reduction rules for \textsc{Connected Dominating Set} in Lokshtanov et al.~\cite{LMS11}, but using a different approach.


\paragraph{\textbf{\emph{Organization of the paper}}.} We start in Section~\ref{sec:prelim} with two preliminary results independent from Theorem~\ref{th:main}. Namely, we first give a simple proof of the {\sf NP}-hardness of \textsc{Total Dominating Set} on planar graphs; this result was already claimed\modif{ in}{} in the literature~\cite{GJ79,Zhu09} but we were not able to find any proof. We then prove that \textsc{Total Dominating Set} satisfies the conditions in order to fit into the framework of Bodlaender et al.~\cite{BFL+09}, and therefore it follows that there {\sl exists} a linear kernel for this problem on planar graphs (and more generally, on graphs of bounded genus). After providing some definitions in Section~\ref{sec:defs}, we describe in Section~\ref{sec:reduction}  our reduction rules for \textsc{Total Dominating Set} when the input graph is embedded in the plane, and in Section~\ref{sec:bound} we prove that the size of a reduced plane \YES-instance is linear in the size of the desired total dominating set. In Section~\ref{sec:concl} we conclude with some directions for further research.


\medskip

We use standard graph-theoretic notation, and we refer the reader to~\cite{Diestel12} for any undefined term. For two integers $i,j$ with $i \leq j$, we use the shortcut $[i,j]$ to denote the set of all integers $\ell$ such that $i \leq \ell \leq j$.

\section{Preliminary results}
\label{sec:prelim}

In Subsection~\ref{sec:NPcomplete} we provide a simple proof of the {\sf NP}-completeness of \textsc{Total Dominating Set} on planar graphs, and in Subsection~\ref{sec:metaKernels} we show that \textsc{Total Dominating Set} satisfies the general conditions of the meta-theorem of Bodlaender et al.~\cite{BFL+09}.

\subsection{{\sf NP}-completeness of Total Dominating Set on planar graphs}
\label{sec:NPcomplete}


\begin{theorem}
\label{th:NPcomplete}
\modif {The decision version of}{} \totdom is {\sf NP}-complete on planar graphs.
\end{theorem}

\begin{proof}
Let \G be a planar graph and $D \subseteq V$.
Checking whether $D$ is a \tds can be clearly done in time $O(n^2)$, so \modif{the decision version of}{} \totdom is in {\sf NP}. We proceed to give a reduction from \textsc{Vertex Cover} on planar graphs, which is known to be {\sf NP}-hard~\cite{GJ79}.

Let $( \G , k)$ be an instance of \textsc{Vertex Cover}, where $G$ is planar. We construct an instance $(\GB , k' =k +2\cdot |V|)$ of \totdom, where each edge of $G$ is replaced with a path with two edges, and a gadget on four vertices is added to each vertex of $G$. Fig.~\ref{img_NP} shows the gadget that replaces an edge \vw. More precisely,
$V'= \{v, v_1,v_2,v_3,v_4 : v \in V \} \cup \{e_{v,w}: \{v,w\} \in E \}= \bigcup_{i \in [1,4]} V_i \cup V \cup V_E $,
where $V_i = \bigcup_{v \in V} \{v_i\}$ and $V_E=\{e_{u,v} : \{u,v\} \in E\}$, and $E'= \{ \{v,e_{v,w}\}, \{w,e_{v,w}\} : \{v,w\}\in E \} \cup \{ \{v, v_1\}, \{v_1, v_2\}, \{v_1, v_3\}, \{v_2, v_4\} : v \in V \} $. Note that if $G$ is planar, then $G'$ is clearly planar as well.

\begin{figure}[htb]
\begin{center}
   \includegraphics[scale=0.85]{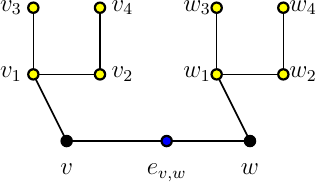}
\end{center}\vspace{-.55cm}
   \caption{Gadget corresponding to an edge \vw. Vertices depicted with $\bullet$ correspond to vertices from $V$, while colored vertices are new.}
   \label{img_NP}
\end{figure}

Let $D$ be a vertex cover in $G$ of size $k$. We define a \tds of $G'$ as $D' = D \cup V_1\cup V_2 $.
We have that $|D'| = k + 2\cdot |V|$.
Vertices from $V_E$ are dominated in $G'$ since $D$ covers all edges of $G$.
Vertices from $V,V_2$, and $V_3$ are dominated by vertices from $V_1$, and
vertices from $V_1$ and $V_4$ are dominated by vertices from $V_2$.

Conversely, let $D'$ be a \tds in $G'$ of size $k'$. We define $D = D' \cap V$.
Edges from $E$ are covered in $G$ since vertices from $V_E$ are dominated by vertices from $D'$. We have $|D| \leq k' - 2 \cdot |V|$, since $V_1$ and $V_2$ are necessarily included in $D'$ (vertices from $V_3$ and $V_4$ need to be dominated by vertices from $V_1$ and $V_2$). The theorem follows.
\end{proof}

\vspace{-.35cm}
\subsection{Meta-kernelization for \textsc{Total Dominating Set}}
\label{sec:metaKernels}


Before proving Lemma~\ref{lem:meta}, we briefly recall some definitions  from Bodlaender et al.~\cite{BFL+09}.

Given an embedded graph $G=(V,E)$, the \emph{radial distance}
between two vertices $x, y$ of $G$ in this embedding is one less than the minimum length of
an alternating sequence of vertices and faces starting from $x$ and ending in $y$, such that
every two consecutive elements of this sequence are incident with each other. Given an embedded graph $G=(V,E)$, a subset $S \subseteq V$, and an integer $r \geq 1$, $R^r_G(S)$ denotes the set of vertices whose radial distance from $S$ is at most $r$ in $G$. Let $\Pi \subseteq \mathcal{G}_g \times \mathbb{N}$ be a parameterized problem defined on graphs of genus $g$. We say that $\Pi$ is \emph{$r$-compact} (for $r \in \mathbb{N}$) if for all instances $(G=(V,E) , k \in \mathbb{N}) \in \Pi$ there exists a set of vertices $S \subseteq V$ and an embedding of $G$ such that $|S| \leq r \cdot k$ and $R^r_G(S) =V$. A problem $\Pi$ is \emph{compact} if there exists an $r$ such that $\Pi$ is $r$-compact. A graph \G is \emph{$t$-boundaried} if it contains a set of $t$ vertices, called $\delta (G)$, labeled from $1$ to $t$. We call $\mathcal{B}_t$ the class of these graphs. The \emph{gluing} of two $t$-boundaried graphs $G_1$ and $G_2$ is the operation $G_1 \oplus G_2$ which identifies vertices from $\delta (G_1)$ and $\delta (G_2)$ with the same label. Two labeled vertices are neighbors in $G_1 \oplus G_2$ if they are neighbors in $G_1$ or in $G_2$. Given a $t$-boundaried graph \G, we define the function $\zeta_G : G'=(V',E') \in \mathcal{B}_t, S' \subseteq V' \mapsto \min\{|S| : S\subseteq V \text{ and } P_\Pi(G \oplus G' , S \cup S') \text{ is true} \}$, where $P_\Pi$ is the language which certificates solutions to problem $\Pi$. (If there is no such $S$, then $\zeta_G$ is undefined.) Assume that $\Pi$ is a MSO minimization problem. We say that $\Pi$  is \emph{(strongly) monotone} if there exists a function $f$ such that for any $t$-boundaried graph $G$ there is a subset of vertices $S$ such that for all $(G',S')$ with $\zeta_G(G',S')$ defined, $P_\Pi(G \oplus G', S \cup S')$ is verified and $|S| \leq \zeta_G(G',S') + f(t) $.

The following theorem is an immediate consequence of~\cite[Theorem 2 and Lemma 12]{BFL+09}.

\begin{theorem}
\label{th:metakernel}
Let $\Pi \subseteq \mathcal{G}_g \times \mathbb{N}$ be a parameterized problem on graphs of genus $g$. If $\Pi$ is monotone and $\Pi$ or $\bar \Pi$ is compact, then $\Pi$ admits a linear kernel.
\end{theorem}

\begin{lemma} \label{lem:meta}
There exists a linear kernel on graphs of bounded genus for \totdom\ parameterized by the solution size.
\end{lemma}

\begin{proof}
By Theorem~\ref{th:metakernel}, we just need to prove that \totdom, parameterized by the solution size, is monotone and compact. Let $G$ be a $t$-boundaried graph 
and let $k\in \mathbb{N}$. By definition of \totdom, if $(G,k)$ admits a solution $D$ then $|D| \leq k$ and $R_G^1(D) = V$. So \totdom\ is clearly compact (for $r=1$). If $D$ is a solution of $G$, then $D$ and $G$ verify the MSO predicate $P(G,D)=\{\forall v \in V, \exists d \in D, adj(v,d) \}$.
Let $D^{(2)}$ be a \tds of minimal size for $G$ and let $D^{(3)} \subseteq N(\delta(G))$ containing a neighbor in $G$ for each vertex from
$\delta(G)$, if such a neighbor exists. We define $D := D^{(2)} \cup D^{(3)} \cup \delta(G)$. Let $(G',D')$ with $G'=(V',E')$ a $t$-boundaried graph, $D'\subseteq V'$, and $\zeta_G(G',D')$ defined. We have that $D \cup D'$ is a \tds for $G \oplus G'$ and that $|D| \leq \zeta_G(G',D') + 2 \cdot t$. Note that if $v \in \delta(G)$ is isolated in $G$ (so, without neighbor from $D^{(3)}$) then it has necessary a neighbor in $D'$, as otherwise $\zeta_G(G',D')$ would be undefined. Thus, \totdom\ is monotone, and the lemma follows.
\end{proof}

\section{Definitions}
\label{sec:defs}

In this section we give necessary definitions for our reduction rules and the analysis of the kernel size. Most of them are inspired from the ones of Alber et al.\modif{ in}{}~\cite{AFN04} for obtaining a linear kernel for \dom on planar graphs. We also use some of the definitions given in~\cite{GST13} for \textsc{Red-Blue Dominating Set}, which were introduced in order to clarify several arguments that were missing in~\cite{AFN04} concerning the properties of regions and region decompositions.


We split the neighborhood of a vertex into three subsets, which intuitively correspond to the layers of ``confinement'' of the vertices inside the neighborhood with respect to the rest of the graph. These three neighborhoods will be used to state Rule~\ref{rgl: Tot seul}.

\begin{definition} \label{def: vois seul}
Let \G be a graph and let $v \in V$. We denote by $N(v) = \{u \in V : \{u,v\} \in E \}$ the neighborhood of $v$. We split $N(v)$ into three subsets:
\begin{itemize}
\item[] $ N_1(v) = \{u \in N(v) : N(u) \setminus (N(v)\cup \{v\}) \neq \emptyset \} $,
\item[]$ N_2(v) = \{u \in N(v)\setminus N_1(v) : N(u) \cap N_1(v) \neq \emptyset \} $, and
\item[] $ N_3(v) = N(v) \setminus (N_1(v) \cup N_2(v))$.
\end{itemize}
For $i,j \in [1,3]$, we denote $N_{i,j} (v) = N_i(v) \cup N_j(v)$.
\end{definition}

The above sets can be intuitively interpreted as follows. The vertices in $N_1(v)$ are adjacent to the rest of the graph and possibly belong to the solution.
The vertices in $N_2(v)$ play the role of a ``buffer'', they can be dominated by $v$ or by some vertex in $N_1(v)$, but they are useless as dominating vertices. 
Finally, the vertices in $N_3(v)$ are isolated and $v$ is the best way to dominate them Hence, if $N_3(v)$ is not empty, we can assume that $v$ is in the solution. In that case, we can safely delete $N_3(v)$ but one vertex; this is done by Rule~\ref{rgl: Tot seul}.

\smallskip

Similarly, we split the neighborhood of a pair of vertices into three subsets. These three neighborhoods will be used to state Rule~\ref{rgl: Tot paire}.

\begin{definition}\label{def: vois paire}
Let \G be a graph and $v,w \in V$. We denote by $N(v,w) = N(v) \cup N(w)$ the neighborhood of the pair $v,w$. We split $N(v,w)$ into three subsets:
\begin{itemize}
\item[] $ N_1(v,w) = \{u \in N(v,w) \mid N(u) \setminus (N(v,w)\cup \{v,w\}) \neq \emptyset \} $,
\item[] $ N_2(v,w) = \{u \in N(v,w)\setminus N_1(v,w) \mid N(u) \cap N_1(v,w) \neq \emptyset \}$, and
\item[] $ N_3(v,w) =  N(v,w) \setminus (N_1(v,w) \cup N_2(v,w))$.
\end{itemize}
For $i,j \in [1,3]$, we denote $N_{i,j}(v,w) = N_i(v,w) \cup N_j(v,w)$.
\end{definition}

Again, the vertices in $\N{1}$ are adjacent to the rest of the graph, \N2 is a kind of ``buffer'', and the vertices in \N3 are isolated and should be dominated by $v$ or $w$. But in this case, we cannot simply delete vertices when $\N3 \neq \emptyset$, because \N3 is not necessarily dominated by $v$ or $w$ in an optimal solution. Indeed, choosing one, two, or three vertices in \N{2,3} can be a better choice than $v$ or $w$. Nevertheless, we will see in Rule~\ref{rgl: Tot paire} that if \N3 is large enough, some vertices can indeed be deleted.

%
%
%

\medskip

Now we proceed to define the concept of \emph{region}, which will play a fundamental role in order to bound the kernel size. Regions are a tool used to decompose a planar graph based on an embedding. In Section~\ref{sec:bound}, Rules~\ref{rgl: Tot paire} and~\ref{rgl: Tot aux} will enable us to bound the number of vertices in a region. 

In order to formally define a region, we need the notion of \emph{confluence}, as defined in~\cite{GST13}. Intuitively, two paths are confluent if they do not cross, but they can possibly merge and split. In other words, the plane can be cut into two parts, each containing one of the paths. For a vertex $u$ in an embedded graph, and neighbors $v_1,v_2,v_3$ of $u$, we use $[v_1<v_2<v_3]_u$ to denote that in the anticlockwise cyclic order around $u$ induced by the embedding of the graph,  vertex $v_2$ comes after $v_1$ and vertex $v_3$ comes after $v_2$.

\begin{definition}\label{def: confluent}
Two simple paths $p_1,p_2$ in  a plane graph $G$ are \emph{confluent} if:
\begin{itemize}
\item[$\bullet$]   they are vertex-disjoint, or
\item[$\bullet$] they are edge-disjoint and for every common vertex $u$, if $v_i,w_i$ are the neighbors of $u$ in $p_i$, for $i \in [1,2]$, it holds that $[v_1< w_1<v_2<w_2]_u$, or
\item[$\bullet$] they are confluent after contracting common edges.
\end{itemize}
\end{definition}

Fig.~\ref{fig: confluent} gives some examples of confluent and non-confluent paths. Let us remark that a path is always confluent with itself. Using the terminology of \cite{GKT15}, the term confluent may be understood as ``sharing the same flow''.

\begin{figure}[htbp]
\begin{center}
  \includegraphics[scale=1]{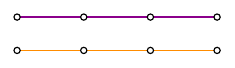}(a)
  \includegraphics[scale=1]{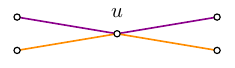}(b)
  \includegraphics[scale=1]{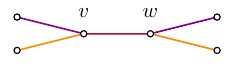}(c)\vspace{.5cm}
  \includegraphics[scale=1]{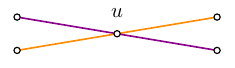}(d)
  \includegraphics[scale=1]{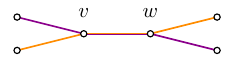}(e)
\caption{\small{ Examples and counterexamples of confluent paths (cf. Definition~\ref{def: confluent}).
  (a) The paths are confluent because they are vertex-disjoint.
  (b) The paths are confluent because vertices are correctly ordered around $u$.
  (c) The paths are confluent because contracting $\{v,w\}$ we obtain the example in (b).
  (d) The paths are not confluent because the order around $u$ is incorrect.
  (e) The paths are not confluent because contracting $\{v,w\}$ we obtain the example in  (d).}}
  \label{fig: confluent}\vspace{-.5cm}
\end{center}
\end{figure}

We now provide our definition of region, which is similar to the one given in~\cite{GST13}. Intuitively, a region is a subset of the plane inducing a subgraph with small diameter, such that the boundary of the region and the boundary of the subgraph match.
In the reduction rules, we will show that the subgraph induced by a region can be ``simulated'' by a small family of gadgets, independently of how the region glues to the rest of the graph. These gadgets correspond to the way a region can intersect a solution.


\newcommand{\bound}{boundary\xspace}

\begin{definition} \label{def: region}
Let \G be a plane graph and let $v,w \in V$ be two distinct vertices.
A $vw$-\emph{region} is a closed subset of the plane such that:
\begin{itemize}
\item[$\bullet$] the boundary of $R$ is formed by two confluent simple $vw$-paths with length at most 3,
\item[$\bullet$] every vertex in $R$ belongs to \N{}, and
\item[$\bullet$] the complement of $R$ in the plane is connected.
\end{itemize}

We denote by $\partial R$ the boundary of $R$ and by $V(R)$ the set of vertices whose images, via the embedding, are in $R$. The \emph{poles} of $R$ are the vertices $v$ and $w$. The \bound paths
are the two $vw$-paths that form $\partial R$.
Given two $vw$-regions $R_1, R_2$ we denote by $R_1 \uplus R_2$ the region defined by $R_1 \cup R_2$, provided that it is indeed a $vw$-region.
\end{definition}

Examples and counterexamples of $vw$-regions are given in Fig.~\ref{fig: region patho}. In the above definition, by the \emph{length} of a path we mean its number of edges.

\begin{figure}[htbp]
\begin{center}
  \includegraphics[scale=1]{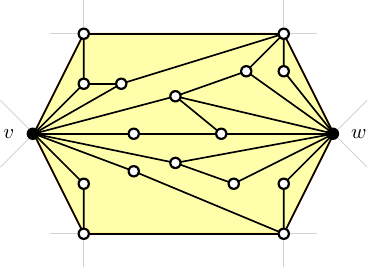}(a)

  \includegraphics[scale=1]{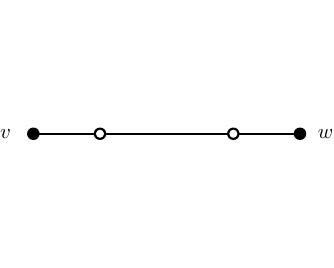}(b)
  \includegraphics[scale=1]{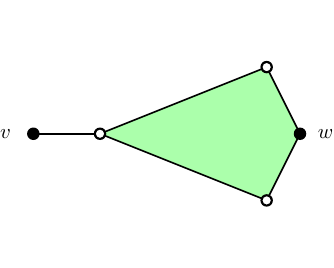}(c)
  \includegraphics[scale=1]{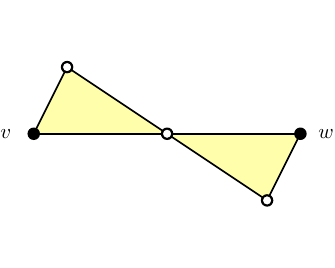}(d)
  \includegraphics[scale=1]{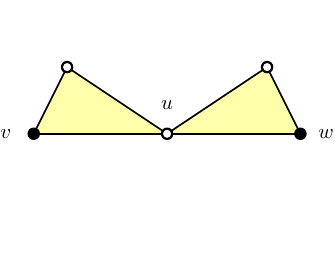}(e)
  \hspace*{-42.5pt}\includegraphics[scale=1]{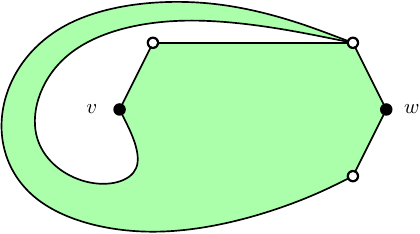}(f)
  \includegraphics[scale=1]{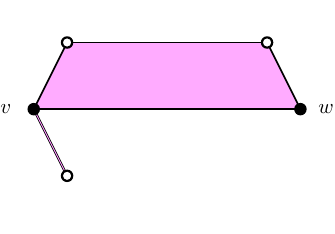}(g)
  \bigskip
  \caption{\small{ Examples and counterexamples of $vw$-regions (cf. Definition~\ref{def: region}).
  (a) A $vw$-region.
  (b)-(c)-(d) Partially degenerate regions: the \bound paths have common edges or common vertices.
  (e)-(f)-(g) Closed sets which are not regions
		because, respectively, the \bound paths are not confluent on $u$,
		because its complementary is disconnected,
		because a part of the boundary is a walk. }}
  \label{fig: region patho}
\end{center}
\end{figure}

	In our kernelization algorithm, the poles of the regions will always be vertices in the total dominating set. The choice of the length of the \bound paths is motivated by the fact that, in a graph whose minimum total dominating sets are of size at least two, every dominating vertex is at distance at most 3 from another dominating vertex.
Note also that the \bound paths of a region are uniquely defined.

	Let us briefly discuss the constraints that we imposed in Definition~\ref{def: region}. First, we ask the poles to be distinct just to avoid loops in the underlying multigraph (see Definition~\ref{def: sous-jacent}). The fact that the complement of a region has to be connected will be used in Proposition~\ref{prop: taille decompo}. Finally, the fact that the two boundary paths are confluent will be used in Proposition~\ref{prop: ext decompo}. 
	Note that the operator $\uplus$ can only be used when the two regions have the same poles and share a \bound path. In that case, the boundary of the new region is formed by the two other paths.

\smallskip

Simple regions are regions in which every vertex is a common neighbor of $v$ and $w$. This special case of region will be used in Rule~\ref{rgl: Tot aux}, and in Propositions~\ref{prop: ext decompo} and~\ref{prop: int region}.

\begin{definition} \label{def: region simple}
A \emph{simple $vw$-region} is a $vw$-region such that:
\begin{itemize}
\item[$\bullet$] its boundary paths have length at most 2, and
\item[$\bullet$] $V(R) \setminus \{v,w\} \subseteq N(v) \cap N(w)$.
\end{itemize}
\end{definition}


\medskip

We now describe how to decompose a plane graph into regions. For this, we will use a \tds as the set of poles of the regions. Since the considered graph is planar and since we choose non-crossing regions, we can use Euler's formula in order to bound the number of regions. 
To this aim, as in~\cite{GST13}, we need to extend the notion of confluence to regions, which we call non-crossing. Intuitively, two regions are non-crossing if their interiors do not intersect and their boundaries are confluent. In other words, they can (partially) share a \bound path, but the plane can still be cut such that each part contains only one region. This property enables us to consider the decomposition as a planar multigraph.

\begin{definition} \label{def: non secant}
 Two regions $R_1,R_2$ are \emph{non-crossing} if:
 \begin{itemize}
 \item[$\bullet$] $(R_1 \setminus \partial R_1) \cap R_2 = (R_2 \setminus \partial R_2)   \cap R_1 = \emptyset $, and
 \item[$\bullet$] the \bound paths of $R_1$ are pairwise confluent with the ones of $R_2$.
 \end{itemize}
 We say that two regions $R_1,R_2$ \emph{cross} on $u$ if $u$ is a common vertex of two non-confluent \bound paths of $R_1$ and $R_2$ such that:
 \begin{itemize}
 \item[$\bullet$] the neighbors of $u$ do not respect the order around $u$ according to Definition~\ref{def: confluent}, or
 \item[$\bullet$]  the regions cross on $u_c$ when an edge $\{u,u'\}$ is contracted into $u_c$.
 \end{itemize}
\end{definition}

Examples and counterexamples of non-crossing regions are given in Fig.~\ref{fig: secant}. Let us make some remarks about this recursive definition. The first condition, which imposes the interiors of the regions to be disjoint, forbids most types of  crossings, but not the one with two degenerate regions illustrated in Fig.~\ref{fig: secant}(f). The second condition, which imposes confluence, forbids also most types of crossing among the boundary paths, but not the one where a region is included in another one, illustrated in Fig.~\ref{fig: secant}(d)-(e).
In the following, we use the terminology ``two regions cross on a vertex'' only when one of the regions is degenerated to a path.

\begin{figure}[htbp]
\begin{center}
  \includegraphics[scale=1]{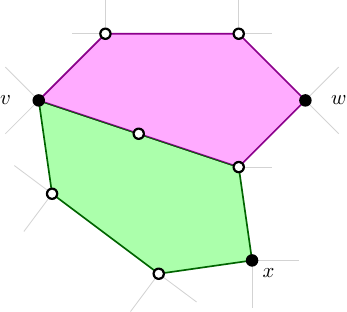}(a)
  \includegraphics[scale=1]{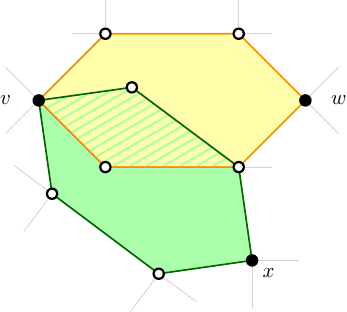}(b)

   \vspace{.5cm}

  \includegraphics[scale=1]{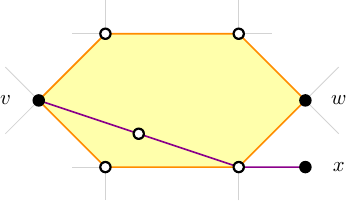}(c)
  \includegraphics[scale=1]{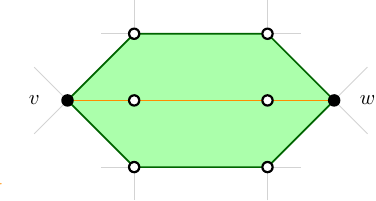}(d)

   \vspace{.5cm}

  \includegraphics[scale=1]{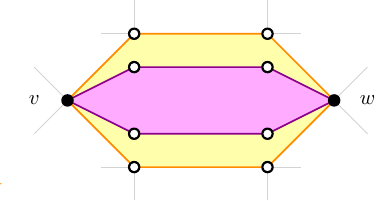}(e)
  \includegraphics[scale=1]{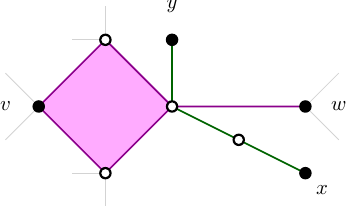}(f)

   \vspace{.5cm}

  \caption{\small{ Examples and counterexamples of non-crossing regions (cf. Definition~\ref{def: non secant}).
  (a) The regions are non-crossing.
  (b)-(c) The regions are crossing because they intersect and the boundaries are not confluent.
  (d)-(e) The regions are crossing because they intersect.
  (f) The regions are crossing because the boundaries are not confluent.  }}
  \label{fig: secant}
\end{center}
\end{figure}

\medskip

	We now provide the definition of region decomposition. Intuitively, a decomposition is a set of regions in the plane such that the poles belong to a fixed set (in our case, it will be a total dominating set) and regions do not cross. In the following we will show that, under some conditions, the number a regions in a decomposition is linear on the number of poles.
	
	\begin{definition}
	Given  a plane graph \G and $D \subseteq V$, a $D$-region decomposition of $G$ is a set $\Re$ of regions with poles in $D$ such that:
	\begin{itemize}
	\item[$\bullet$] for any $vw$-region $R \in \Re$, it holds that  $D \cap V(R) = \{v,w \}$, and
	\item[$\bullet$] all regions are pairwise non-crossing.
	\end{itemize}
	We denote $V(\Re) = \bigcup_{R\in \Re} V(R)$.
	\end{definition}

 Note that a region decomposition noes not necessarily contain all the vertices of the plane graph,  as some vertices may remain uncovered.

\smallskip

We are particularly interested in \emph{maximal} region decompositions, namely, region decompositions that contain a maximal number of vertices with a minimal number of regions. Intuitively, a decomposition is maximal if no vertex can be added, neither by adding a region, nor by extending a region, and, subject to the latter conditions, no pair of regions can be unified. In the following we always assume that the considered decomposition are maximal, which is necessary in order to bound the number of regions.

\begin{definition}
A $D$-region decomposition $\Re$ is \emph{maximal} if:
\begin{itemize}
\item[$\bullet$] for every region $R \notin \Re$ such that $V(\Re) \subsetneq V(\Re \cup \{R\})$,
      $\Re \cup \{R\}$ is not a $D$-region decomposition,
\item[$\bullet$] for every region $R \notin \Re$ and $R' \in \Re$ such that $V(R') \subsetneq V(R)$,
          $\Re \cup \{ R\} \setminus \{R'\}$ is not a $D$-region decomposition, and
\item[$\bullet$] for every two regions $R_1,R_2 \in \Re$, the subset of the plane $R_1 \cup R_2$ is not a region.
\end{itemize}
\end{definition}

 Note that a maximal region decomposition does not necessarily contain a maximum number of vertices. 
 Moreover, the size of the regions is not taken into account in the above definition. Note also that for every vertex set $D$ there exists a maximal $D$-region decomposition, which can be constructed with a simple greedy algorithm \cite{AFN04} (see also~\cite{GST13}). It can be easily checked that such an algorithm also works to build a maximal {\sl simple} region decomposition, which is defined in the natural way.

\medskip

We now provide some tools in order to show that, when $D$ is a \tds,  the number of regions in a  maximal $D$-region decomposition is linear on $|D|$ (see the proof of Proposition~\ref{prop: taille decompo}). Intuitively, we will consider the region decomposition as a multigraph with nice properties and then apply Euler's formula.

\begin{definition}\label{def: thin}
A planar multigraph $M$ is \emph{thin} if there exists an embedding of $M$ such that,
for any pair of edges $e_1,e_2$ with common endvertices, each of the two subsets of the plane delimited by the images of $e_1,e_2$ contains the image of a vertex.

\end{definition}

A thin multigraph satisfies Euler's formula (see~\cite{GST13} for a proof). That is, if $M=(V,E)$ is thin then $|E| \leq 3 |V| - 6$. Recall also that (see~\cite{Diestel12}) for planar bipartite graphs (in particular, stars), Euler's formula can be strengthened to $|E| \leq 2 |V| - 4$.

\medskip

Given a $D$-region decomposition $\Re$ of a plane graph $G$, we consider a multigraph $G_\Re$ containing an edge for each region of $\Re$. This graph is formally defined as follows.

\begin{definition} \label{def: sous-jacent}
Let \G be a plane graph, let $D \subseteq V$, and let $\Re$ be a $D$-region decomposition of $G$. The \emph{underlying
multigraph} $G_\Re = (V_\Re, E_\Re)$ of $\Re$ is such that $V_\Re = D$ and there is an edge $\{v,w\} \in E_\Re$ for each $vw$-region of $\Re$.
\end{definition}

Since the regions are non-crossing, it can be shown that $G_\Re$ is a planar multigraph. Moreover, when $D$ is a \tds and $\Re$ is maximal, it holds that $G_\Re$ is thin. Both facts are proved in the proof of Proposition~\ref{prop: taille decompo}.

\section{Reduction rules}

\label{sec:reduction}

The kernelization algorithm will consist in an exhaustive application of three rules, until the instance becomes \emph{reduced}. Mainly, we consider two rules: the first one is applied on a vertex, and the second one on a pair of vertices. We also have a third auxiliary rule that can be seen as a complement of the second one, in order to obtain a better constant. For applying this rule we need to assume that the graph is embedded.
When we bound the kernel size, the first rule enables us to show that there are few vertices outside the regions, while the second and the third ones allow to prove that there are few vertices inside the regions. In Subsection~\ref{ssec:reduction1} we state the rule for a single vertex, and in Subsection~\ref{ssec:reduction2} we state the rules for a pair of vertices. Finally, in Subsection~\ref{ssec:complexite}  we analyze the time complexity of the reduction procedure.

\medskip
Before starting with the reduction rules, let us formally define what we consider to be a \emph{reduced} graph, which coincides with the definition given in~\cite{GST13}.

\begin{definition}\label{def:reduced}
A graph $G$ is reduced under a set of rules if either none of these
rules can by applied to $G$, or the application of any of them creates a graph isomorphic to $G$.
\end{definition}

We simply say that a graph is \emph{reduced} if it reduced under the
whole set of reduction rules that we will define, namely Rules~\ref{rgl: Tot seul},~\ref{rgl: Tot paire}, and~\ref{rgl: Tot aux}.
We would like to point out that the above definition differs from the usual definition of
reduced graph in the literature, which states that a graph is reduced if the corresponding reduction rules cannot be applied anymore. We diverge from this definition because, for convenience, we will define reduction rules that could be applied \emph{ad infinitum} to the input graph, such as Rule~\ref{rgl: Tot seul} or some cases of Rule~\ref{rgl: Tot paire}.
The reduction rules that
we will define are all local and concern the neighborhood of at most two vertices, which is replaced
with gadgets of constant size. Therefore, from an algorithmic point of view, in order to know when a graph is reduced, the fact whether the original and the modified graph are isomorphic or not can be
easily checked locally in constant time.

\subsection{Reduction rule for a single vertex}
\label{ssec:reduction1}


If we know, somehow, that a vertex has to be a dominating vertex, then we can replace part of its neighborhood with some gadget forcing the vertex to be dominating. To this aim, we need to split the neighborhood of a vertex as in Definition~\ref{def: vois seul}.


\begin{rgl} \label{rgl: Tot seul}
Let \G be a graph and let $v \in V$. If $|N_3(v)| \geq 1 $:
\begin{itemize}
\item[$\bullet$] remove $N_{2,3}(v)$ from $G$, and
\item[$\bullet$] add a vertex $v'$ and an edge $\{v,v'\}$.
\end{itemize}
\end{rgl}

\begin{figure}[h]
 \begin{center}
 \vspace{-.2cm}
 \includegraphics[scale=1]{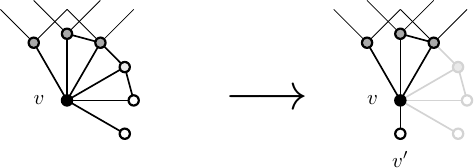}
  \vspace{-.2cm}
 \caption{\small{Application of Rule~\ref{rgl: Tot seul} on a vertex $v$. The rule removes $N_{2,3}(v)$ and adds the vertex $v'$. }}
 \label{fig: Tot seul}
 \end{center}
\end{figure}

An application of Rule~\ref{rgl: Tot seul} is illustrated in Fig.~\ref{fig: Tot seul}. Note that choosing some vertices from $N_{2,3}(v)$  as dominating (instead of $v$) is not the best choice, as $v$ dominates more vertices (since $N(v) \supseteq N(N_{2,3}(v))$). Moreover, any dominating vertex from $N_{2,3}(v)$ has to be dominated by a vertex that also dominates $v$ (as we are looking for a {\sl total} dominating set). Therefore, we can assume that the considered \tds  contains $v$. The added gadget $v'$ simulates $N_3(v)$. On the one hand, it forces $v$ to be in the solution. On the other hand, if $v$ were dominated by some vertex of $N_3(v)$, it can be dominated by $v'$.

\begin{lemma}\label{lem: Tot seul correct}
Let \G be a graph and et $v \in V$. If $G'$ is the graph obtained by the application of Rule~\ref{rgl: Tot seul} on $v$, then $G$ has a \tds of size $k$ if and only if $G'$ has one.
\end{lemma}

\begin{proof}
From a \tds of $G$, we construct a \tds of $G'$ with smaller or equal size, and vice versa.

Let $D$ be a \tds of $G$ of size $k$. We can assume that $v\in D$. Indeed, assume that $v \notin D$. According to Rule~\ref{rgl: Tot seul}, $N_3(v) \neq \emptyset$, so there exists a dominating vertex $u \in N_{2,3}(v)$ and another one $w \in N(u)$. Since $N(u) \subseteq N(v)$, we can replace $u$ with $v$ in $D$, and $v$ is dominated by $w$. Hence, we can assume $v \in D$. We define $D'$ from $D$ as follows:
if Rule~\ref{rgl: Tot seul} has removed dominating vertices, we set  $D' = D \setminus N_{2,3}(v) \cup \{v'\}$, otherwise $D' =D$. If the rule has removed a dominating vertex $u \in D$, then  $N(u) \setminus \{v\}$ is dominated by $v$ in $G'$. Moreover, $v'$ is dominated by $v$, and $v$ is dominated by $v'$. Otherwise, $v'$ is dominated by $v$ and $v$ is dominated by a vertex from $N_1(v)$.
In both cases, $|D'| \leq |D|$.

Conversely, let $D'$ be a \tds of $G'$ of size $k$. It holds that $v \in D'$ because $v'$ has to be dominated. We define $D$ from $D'$ as follows: if $v' \in D'$ then $D = D' \cup \{ u \} \setminus \{ v' \}$, with $u \in N_3(v)$, and otherwise $D=D'$. Note that the rule has removed only neighbors of $v$ that are dominated by $v$ in $G$. If $v'\in D'$, then $v$ is dominated by $u$. Otherwise, $v$ is dominated in $G$ by some vertex from $N_1(v)$.
In both cases, $|D| = |D'|$.
\end{proof}


\subsection{Reduction rules for a pair of vertices}
\label{ssec:reduction2}

We now describe the rules for a pair of vertices. Similarly as before, if we know that one of the two vertices has to be a dominating vertex, we can replace a part of their neighborhood with some gadget forcing that vertex to be in the solution. However, in some cases, the two vertices do not suffice to totally dominate the graph, and we need a more involved gadget. In order to define this gadget, we consider all the possible ways to dominate the neighborhood using few vertices, hence guaranteeing that the reduction preserves the existence of solutions. To this aim, we need to split the neighborhood of a pair of vertices as in Definition~\ref{def: vois paire}.

It can be observed that, in the worst case, at most four vertices are needed to dominate \N3, namely $v$, $w$, and two vertices to dominate $v$ and $w$, respectively. Hence, we want the reduction to  preserve the solutions using one, two, or three vertices (which could be a better choice). To this aim, we will consider all the sets of at most three vertices, and the reduction will remove as many vertices as possible, while preserving the sets dominating \N3. We now formalize the notion of a set that dominates \N3 with the \N3-\tds. Then we precise which \N3-\tds we want to consider.

\begin{definition}\label{def: N3-TDS} 
Let \G be a graph and let $S \subseteq V$. A $S$-\tds $\tilde{D} \subseteq S \cup N(S)$ is a set dominating $S$ such that the induced graph $G[\tilde D]$ contains no isolated vertex from $S$ (but vertices from $N(S)$ may be isolated).

Let $v,w \in V$. We consider the \N3-\tds of size at most 3, and we denote:

\begin{tabbing}
$\Dvw$~ \= $ = \{ \tilde D \subseteq N_{2,3}(v,w)            \mid N_3(v,w) \subseteq \bigcup_{v \in \tilde D} N(v),\ |\tilde D| \leq 3                  \}$,\\
$\Dv$   \> $ = \{ \tilde D \subseteq N_{2,3}(v,w) \cup \{v\} \mid N_3(v,w) \subseteq \bigcup_{v \in \tilde D} N(v),\ |\tilde D| \leq 3,\ v \in \tilde D \}$, and\\
$\Dw$   \> $ = \{ \tilde D \subseteq N_{2,3}(v,w) \cup \{w\} \mid N_3(v,w) \subseteq \bigcup_{v \in \tilde D} N(v),\ |\tilde D| \leq 3,\ w \in \tilde D \}$.
\end{tabbing}
%
We also denote  $\bigcup \Dv = \bigcup_{D \in \Dv} D$ and $\bigcup \Dw = \bigcup_{D \in \Dw} D$.
\end{definition}

Note that the set \Dv (resp. \Dw) contains the sets of vertices that dominate \N3 using $v$ (resp. $w$) plus at most two other vertices. The set \Dvw contains the sets that can dominate \N3
 with neither $v$ nor $w$. We are now ready to state Rule~\ref{rgl: Tot paire}.


\begin{rgl} \label{rgl: Tot paire}
Let \G be a graph and let $v,w \in V$. If $\Dvw = \emptyset$:

\begin{enumerate}
\item if $\Dv = \emptyset$ and $\Dw = \emptyset$: \label{cas: tot et}
  \begin{itemize}
  \item[$\bullet$] remove $N_{2,3}(v,w)$,
  \item[$\bullet$] add vertices $v',w'$ and edges $\{v,v'\}, \{w,w'\}$,
  \item[$\bullet$] if there is a common neighbor of $v$ and $w$ in $N_{2,3}(v,w)$,
        add a vertex $y$ and edges $\{v,y\}, \{w,y\}$;
  \end{itemize}

\medskip
\item if $\Dv \neq \emptyset$ and $\Dw \neq \emptyset$: \label{cas: tot ou}
  \begin{itemize}
    \item[$\bullet$]  remove $(N_{2,3}(v,w) \cap N(v) \cap N(w)) \setminus (\bigcup\Dv \cup  \bigcup\Dw) $,
  \item[$\bullet$] for all distinct vertices $d,d' \in (\bigcup\Dv \cup \bigcup\Dw) \cap (N(v) \cap N(w))$,
        if $N(d) \setminus (N(v)\cap N(w)) \subseteq N(d') \setminus (N(v)\cap N(w))$,
        remove $d$,
  \item[$\bullet$] if a vertex has been removed,
        add a vertex $y$ and edges $\{v,y\}, \{w,y\}$;
  \end{itemize}

\medskip
\item if $\Dv \neq \emptyset$ and $\Dw = \emptyset$: \label{cas: tot v}
  \begin{itemize}
  \item[$\bullet$] remove $(N_{2,3}(v,w) \cap N(v)) \setminus \bigcup\Dv$,
  \item[$\bullet$] for all distinct vertices $d,d' \in \bigcup\Dv \cap N(v)$,
        if $N(d) \setminus N(v) \subseteq N(d') \setminus N(v)$,
        remove $d$,
  \item[$\bullet$] if a vertex has been removed,
         add a vertex $v'$ and an edge $\{v,v'\}$;
  \end{itemize}

\medskip
\item if $\Dv = \emptyset$ and $\Dw \neq \emptyset$: \label{cas: tot w}
  \begin{itemize}
  \item[$\bullet$] do symmetrically to Case~\ref{cas: tot v}.
  \end{itemize}

\end{enumerate}
\end{rgl}

\begin{figure}
 \begin{center}
 \includegraphics[scale=0.9]{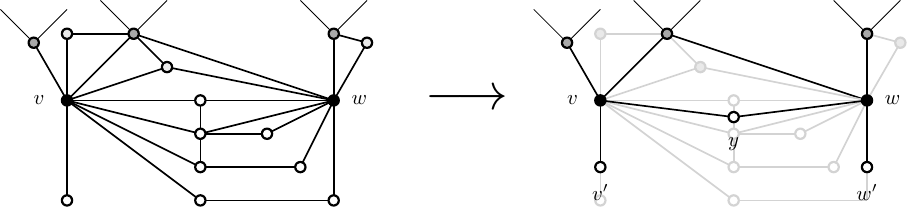}(a)\vspace{.5cm}
 \includegraphics[scale=0.9]{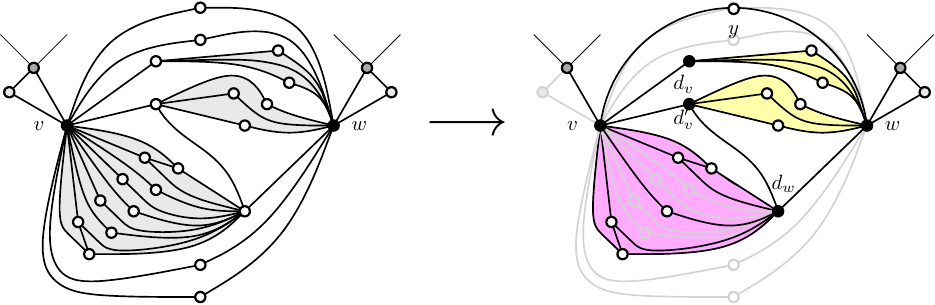}(b)\vspace{.5cm}
 \includegraphics[scale=0.9]{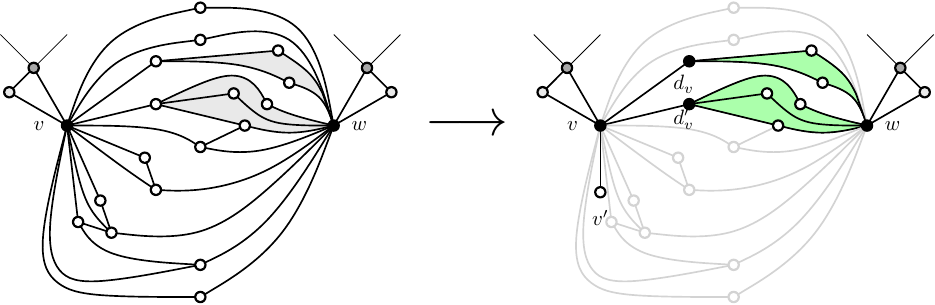}(c)
\bigskip
 \caption{\small{Application of Rule~\ref{rgl: Tot paire} on two vertices $v,w$ and of Rule~\ref{rgl: Tot aux} if needed (see page~\pageref{rgl: Tot aux}).
  (a) In Case~\ref{cas: tot et}, part of $N(v,w)$ is removed and vertices $v',w',y$ are added.
  (b) In Case~\ref{cas: tot ou}, $N(v) \cap N(w)$ is removed and vertex $y$ is added.
      Then Rule~\ref{rgl: Tot aux} reduces a $vd_w$-simple region.
  (c) In Case~\ref{cas: tot v},  part of $N(v)$ is removed and vertex $v'$ is added.
We have represented the simple regions which have to be reduced by Rule~\ref{rgl: Tot aux} in order to obtain the bound on the kernel size.}}
 \label{fig: Tot paire}
 \end{center}
\end{figure}


Examples of applications of Rule~\ref{rgl: Tot paire} can be found in Fig.~\ref{fig: Tot paire}. Let us briefly discuss the distinct cases of application of the rule.

In Case~\ref{cas: tot et}, we can assume that $v$ and $w$ are dominating vertices. Then we can remove \N{2,3}, which is not necessary for domination and which is dominated by $v$ and $w$. In Case~\ref{cas: tot ou}, we can assume that at least one of $v,w$ is in the solution, but we do not know which one. Then we simplify $N(v) \cap N(w)$ which can be dominated by any of $v$ and $w$. In Case~\ref{cas: tot v} (resp. Case~\ref{cas: tot w}) we can assume that $v$ (resp. $w$) is in the solution, and then we simplify $N(v)$ (resp. $N(w)$), which is dominated by $v$ (resp. $w$).

In any case, when we simplify the neighborhood, we preserve the vertices that belong to an \N3-\tds. The gadget $v'$ (or symmetrically, $w'$) forces to choose $v$ (when $\Dw = \emptyset$). The gadget $y$ let us choose $v$ or $w$, but forcing at least one of them. Moreover, $y$ ensures that $v,w$ can be dominated with only one vertex (the distance between $v$ and $w$ does not increase).

In the case where we do not simplify the neighborhood $N(w)$ (or symmetrically, $N(v)$), that is, when $\Dv \neq \emptyset$, there may remain many vertices which cannot be dominating vertices but which are necessary to preserve the \N3-\tds. 
 For this reason, we need Rule~\ref{rgl: Tot aux}. Intuitively, we want to apply Rule~\ref{rgl: Tot aux} on $w$ (or symmetrically, to $v$) and the vertices of an \N3-\tds in \Dv, in order to reduce the neighborhood $N(w)$ that has not be modified by Rule~\ref{rgl: Tot paire}.

\medskip

Before proving the correctness of Rule~\ref{rgl: Tot paire}, we need some facts about the  properties of a graph obtained by the application of Rule~\ref{rgl: Tot paire}. More precisely, they show that we can assume that some vertices belong to the solution.
The first fact states that if Rule~\ref{rgl: Tot paire} can be applied, then we can assume that $v$ or $w$ belong to the \tds of $G$ we are considering.

  \begin{fact}\label{fait 1}
  Let \G be a graph, let $v,w \in V$, and let $G'$ be the graph obtained by the application of Rule~\ref{rgl: Tot paire} on $v,w$.
  If $\Dvw = \emptyset$, then $G$ has a solution if and only if it has a solution containing at least one of the two vertices $v,w$.
  \end{fact}

  \begin{proof}
By assumption, $\Dvw = \emptyset$, and thus any \tds of $G$ has to contain $v$ or $w$, or at least four vertices from \N{2,3}. In the latter case, the four vertices can be replaced with $v$, $w$, and two neighbors of $v$ and $w$, respectively. Therefore, we can consider only solutions containing $v$ or $w$.
\end{proof}

  The second fact states that if no set of the form $\{w\}, \{w,u\}, \{w,u,u'\}$ with $u,u' \in \N{2,3}$, can dominate \N3, then we can consider in the original graph only solutions containing $v$.


  \begin{fact}\label{fait 2}
  Let \G be a graph, let $v,w \in V$, and let $G'$ be the graph obtained by the application of Rule~\ref{rgl: Tot paire} on $v,w$.
  If $\Dw = \emptyset$ (resp. $\Dv= \emptyset$) and $\Dvw = \emptyset$, then $G$ has a solution if and only if it has a solution containing $v$ (resp. $w$). Moreover $(N(v) \setminus N(w)) \cap N_3(v,w) \neq \emptyset$ (resp. $(N(w) \setminus N(v)) \cap N_3(v,w) \neq \emptyset$).
  \end{fact}

  \begin{proof}
By assumption, $\Dw = \emptyset$, hence no set of the form $\{w\}, \{w,u\}, \{w,u,u'\}$ with $u,u' \in \N{2,3}$, can dominate \N3. In particular, $v$ has a neighbor in $G$ which is not a neighbor of $w$. Since $\Dvw = \emptyset$, any \tds of $G$ has to contain $v$ or at least four vertices. In the latter case, the four vertices can be replaced with $v$, $w$ and two neighbors of $v$ and $w$, respectively. Therefore, we can assume that $v$ belongs to the solution.
  \end{proof}

Conversely, the third fact states  that if we cannot use $w$  to dominate in the original graph, then we can assume that $v$ belongs to the solution in the reduced graph.

   \begin{fact}\label{fait 3}
  Let \G be a graph, let $v,w \in V$, and let $G'$ be the graph obtained by the application of Rule~\ref{rgl: Tot paire} on $v,w$.
   If $\Dw = \emptyset$ (resp. $\Dv= \emptyset$) and $\Dvw = \emptyset$, then $G'$ has a solution if and only if it has a solution containing $v$ (resp. $w$).
   \end{fact}

  \begin{proof}
If Rule~\ref{rgl: Tot paire} has not modified the graph then, according to Fact~\ref{fait 1}, $G'=G$ has a solution containing $v$. Otherwise, Rule~\ref{rgl: Tot paire} has added $v'$ such that $N(v') = \{v\}$, hence any \tds of $G'$ has to contain $v$ in order to dominate $v'$.
  \end{proof}

   The fourth fact states that if a vertex in \N{2,3} is a dominating vertex, then we can assume that this vertex belongs to a solution. In other words, the vertices that do not belong to an \N3-\tds are never a good choice.

   \begin{fact}\label{fait 4}
  Let \G be a graph, let $v,w \in V$, and let $G'$ be the graph obtained by the application of Rule~\ref{rgl: Tot paire} on $v,w$.
   If $\Dv \neq \emptyset$ (resp. $\Dw \neq \emptyset$) and $\Dvw = \emptyset$, then $G$ has a solution if and only if it has a solution containing no vertices form $\N{2,3} \setminus (\bigcup\Dv\cup\bigcup\Dw)$.
   \end{fact}

   \begin{proof}
Let $u \in \N{2,3} \setminus (\bigcup\Dv\cup\bigcup\Dw)$. Since $\Dvw = \emptyset$, any \tds of $G$ containing $u$ has to contain at least four vertices from $N(v,w) \cup \{v,w\}$. In that case, the four vertices can be replaced either with a set from $\Dv \cup \Dw$, or with four vertices from $\{v,w\} \cup \N1$.
   \end{proof}

We are now ready to prove the correctness of Rule~\ref{rgl: Tot paire}.

\begin{lemma}\label{lem: Tot paire correct}
Let \G be a graph and let $v,w \in V$. If \GB is the graph obtained by the application of Rule~\ref{rgl: Tot paire} on $v,w$, then $G$ has a \tds of size $k$ if and only if $G'$ has one.
\end{lemma}

   \begin{proof}
We start considering a \tds $D$ of $G$ in order to build a \tds $D'$ of $G'$, with $|D'| \leq |D|$. We prove independently each case of the rule. Let $D$ be a \tds of $G$.

\begin{enumerate}
\item $\Dv = \emptyset$ and $\Dw = \emptyset$. We build $D'$ starting with $D \cap V'$, and depending on whether the rule removes a vertex from $D \cap N(w) \setminus N(v)$, from $D \cap N(v) \setminus N(w)$, or from $D \cap N(v) \cap N(w)$. For each case, we add $v'$, $w'$, or $y$, respectively. By Fact~\ref{fait 2}, we can assume that $v,w \in D$. On the one hand, $v'$, $w'$, and $y$ are dominated in $G'$ by $v$ or $w$. On the other hand, if the rule has removed a vertex $u \in D$, then $u \in \N{2,3}$ and, by definition, $N(u) \subseteq N(v,w)$.
Hence $N(u) \setminus \{v,w\}$ is dominated in $G'$ by $v$ or $w$,
      the vertices $v,w$ are dominated, either by  $v',w',y$ (if they have been added), or by \N1 (otherwise).
Finally, all other vertices of $G'$ are dominated in the same way than in $G$. It follows that $D'$ is a \tds of $G'$ with $|D'| \leq |D|$.

\medskip
\item $\Dv \neq \emptyset$ and $\Dw \neq \emptyset$. We build $D'$ starting with $D \cap V'$ and, for each vertex $d \in D$ removed by the rule, we add a vertex $d' \in N(v) \cap N(w)$ such that $N(d) \setminus (N(v)\cap N(w)) \subseteq N(d') \setminus (N(v)\cap N(w))$. By Fact~\ref{fait 4}, $d \in \bigcup\Dv\cup\bigcup\Dw$, hence $d'$ exists. By Fact~\ref{fait 1}, we can assume that $v \in D$ or $w \in D$. On the one hand, $y$ (if it has been added) is dominated in $G'$ by $v$ or $w$. On the other hand, if the rule has removed a vertex $d \in D$, then there exists $d' \in D'$ such that $N(d) \setminus (N(v)\cap N(w)) \subseteq N(d') \setminus (N(v)\cap N(w))$, and then $N(d)$ is dominated in $G'$ by $d'$. It follows that $D'$ is a \tds of $G'$ with $|D'| \leq |D|$.

\medskip
\item $\Dv \neq \emptyset$ and $\Dw = \emptyset$ (Case~\ref{cas: tot w} is symmetric).
Again, we build $D'$ starting with $D \cap V'$ and, for each vertex $d \in D$ removed by the rule, we add a vertex $d' \in N(v) \cap N(w)$ such that $N(d) \setminus N(v) \subseteq N(d') \setminus N(v)$. According to Fact~\ref{fait 4}, $d \in \bigcup\Dv$, hence $d'$ exists. By Fact~\ref{fait 2}, we can assume that $v \in D$.  On the one hand, $v$ (if it has been added) is dominated in $G'$ by $v$. On the other hand, if the rule has removed a vertex $d \in D$, then there exists $d' \in D'$ such that $N(d) \setminus N(v) \subseteq N(d') \setminus N(v)$, and then $N(d)$ is dominated in $G'$ by $d'$. It follows that $D'$ is a \tds of $G'$ with $|D'| \leq |D|$.
\end{enumerate}

Conversely, we now consider a \tds $D'$ of $G'$ in order to build a \tds $D$ of $G$ with $|D| \leq |D'|$. Again, we prove independently each case of the rule. 

\begin{enumerate}
\item
We build $D$ starting with $D' \setminus \{v',w',y\}$, and depending on whether $v' \in D'$, $w' \in D'$, or $y \in D'$, we add, respectively, $u_v \in N(v) \setminus N(w)$, $u_w \in N(w) \setminus N(v)$, $u \in N(v) \cup N(w)$. Such vertices exist by Fact~\ref{fait 2}. By Fact~\ref{fait 3}, $v,w \in D'$, and so $v,w \in D$. On the one hand, the rule removes only neighbors of $v$ or $w$ which are dominated in $G$ by $v$ or $w$. On the other hand, the vertices $v,w$ are dominated in $G$, either by $u_v$, $u_w$, or $u$ (if they have been added), or by \N1 (otherwise). It follows that $D$ is a \tds of $G$ with $|D| = |D'|$.

\medskip
\item We build $D$ starting with $D' \setminus \{y\}$ and, if $y \in D'$, we add a vertex $u \in N(v) \cup N(w)$, which is guaranteed to exist.  By Fact~\ref{fait 3}, $v\in D'$ or $w \in D'$, and so $v\in D$ or $w \in D$. On the one hand, the rule removes only common neighbors of $v$ and $w$, which are dominated in $G$ by $v$ or $w$. On the other hand, $v$ and $w$ are dominated in $G$ either by $u$ (if it has been added) or by non-modified vertices (otherwise). It follows that $D$ is a \tds of $G$ with $|D| = |D'|$.

\medskip
\item  (Case~\ref{cas: tot w} is symmetric.) We build $D$ starting with $D' \setminus \{v'\}$ and, if $v' \in D'$, we add a vertex $u \in N(v) \cup N(w)$,  which is guaranteed to exist.
By Fact~\ref{fait 3}, $v\in D'$, and so $v\in D$. On the one hand, the rule removes only neighbors of $v$, which are dominated in $G$ by $v$. On the other hand, $v$ is dominated either by $u$ (if it has been added) or by non-modified vertices (otherwise). It follows that $D$ is a \tds of $G$ with $|D| = |D'|$.
\end{enumerate}\vspace{-.2cm}
 \end{proof}

We now describe the auxiliary rule, namely Rule~\ref{rgl: Tot aux}, which completes Rule~\ref{rgl: Tot paire}. Intuitively, this rule reduces the part of \N3 that has not been removed by Rule~\ref{rgl: Tot paire}. To state this rule, we need to assume that the graph is embedded in the plane because we use simple regions.

\begin{rgl} \label{rgl: Tot aux}
	Let \G be a plane graph, let $v,w \in V$ a pair of vertices, and let $R$ be a simple $vw$-region. If $|V(R) \setminus \{v,w\}| \geq 5$:
	\begin{itemize}
	\item[$\bullet$] remove $V(R \setminus \partial R)$, and
	\item[$\bullet$] add a vertex $z$ and edges $\{v,z \} , \{w,z \}$.
	\end{itemize}
	\end{rgl} 

\begin{figure}
 \begin{center}
  \includegraphics[scale=0.9]{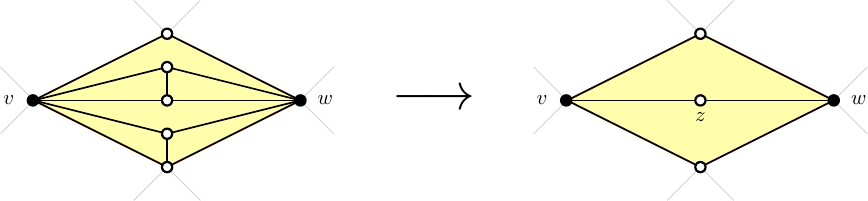}\vspace{-.05cm}
 \caption{\small{ Application of Rule~\ref{rgl: Tot aux} on a simple $vw$-region (see Definition~\ref{def: region simple}).
  }}
 \label{fig: Tot aux}\vspace{-.1cm}
 \end{center}
\end{figure}

An example of the application of Rule~\ref{rgl: Tot aux} is illustrated in Fig.~\ref{fig: Tot aux}. Note that the above rule is specifically designed to bound the size of simple regions.
This auxiliary rule is not strictly necessary, since Rule~\ref{rgl: Tot paire} can also by used to bound the size of simple region (as done in \cite{AFN04}). But if we did so, the bound on the size will be nine instead of five, which would result in a significantly worse kernel size. However, we are convinced that an improved version of Rule~\ref{rgl: Tot paire} can be described in order to obtain the bound of five. Indeed, when we define \N3-\tds (Definition~\ref{def: N3-TDS}) we observe that four vertices are always enough to totally dominate \N3; but assuming that $v,w$ are at distance two, then three vertices are enough. Next, considering the fact that an \N3-\tds has to dominate all \N3 and not only vertices in the simple region, we think that the bound of five can be achieved. Such a rule would be harder to formulate, but would avoid to deal with the embedding of the graph during the reduction procedure. For the sake of simplicity, we decided to state Rule~\ref{rgl: Tot aux} independently.

Observe also that in this rule, two vertices on the boundary of $R$ play the role of \N1 (they have neighbors outside the region), at most two adjacent vertices take the role of \N2 (they are a ``buffer''), and the remaining vertices of $R$ play the role of \N3 (they are isolated from the outside).

\begin{lemma}\label{lem: Tot aux correct}
Let \G be a plane graph, let $v,w \in V$, and let $R$ be a simple $vw$-region.
If $G'$ is the graph obtained by the application of Rule~\ref{rgl: Tot aux} on $R$, then $G$ has a \tds of size $k$ if and only if $G'$ has one.
\end{lemma}

\begin{proof}
Let $D$ be a \tds of $G$, and we construct a \tds $D'$ of $G'$ as follows:
  if the rule has removed a vertex from $D$, then $D' = D \setminus (\N3 \cap V(R)) \cup \{z\} $, and  otherwise $D'=D$. Since $|V(R)| \geq 5$, we can assume that $v$ or $w$ is a dominating vertex. Indeed, in order to dominate $V(R \setminus \partial R)$, we need $v$ or $w$ or at least one vertex strictly inside $R$. In the latter case,  this vertex can be replaced with $v$ or with $w$. 
Assume without loss of generality that $v \in D$. If a vertex $u \in D$ is removed by the rule, then the neighbors of $u$ are dominated by $v$ in $G'$, except $v$ and $w$ which are dominated by $z$. Moreover, $z$ is dominated by $v$. If the rule removes no vertex, then $v,w$ are dominated in the same way than in $G$.
  It follows that $D'$ is a \tds of $G'$ with $|D'| \leq |D|$.

Conversely, let $D'$ be a \tds of $G'$, and we construct a \tds $D$ of $G$ as follows:
  if $z \in D'$ then, $D= D' \setminus \{ z \} \cup \{u\}$ with  $u$ strictly in $R$, and
  otherwise $D=D'$. Since $z$ must be dominated, necessarily $v \in D'$ or $w \in D'$. Assume without loss of generality that $v \in D'$. The rule removes only vertices from the set $V(R \setminus \partial R)$, which is dominated by $v$ in $G$. Moreover, $v,w$ are dominated by $u$ if it has been added in $D$, or in the same way than in $G$, otherwise.
  It follows that $D$ is a \tds of $G$ with $|D| = |D'|$.
\end{proof}

\subsection{Complexity of the reduction procedure}
\label{ssec:complexite}

We now show that the reduction procedure runs in polynomial time with respect to the size of the graph

\begin{lemma}\label{lem: cplx}
A plane graph $G$ can be reduced by Rules~\ref{rgl: Tot seul},~\ref{rgl: Tot paire}, and~\ref{rgl: Tot aux} in time $O(|V(G)|^3)$.

\end{lemma}

\begin{proof}
We denote $ n = |V(G)|$ and $d(v) = |N(v)|$ the degree of a vertex $v \in V(G)$. Rule~\ref{rgl: Tot seul} is applied in time $O(d(v))$. Indeed, the construction of the three neighborhoods  is done in  $O(d(v))$. Let us first deal with $N_1(v)$. For each vertex $u \in N_1(v)$, we explore is neighbors until we find a vertex not adjacent to $v$. If such a vertex exists, $u$ belongs to $N_1(v)$.
Hence,  in order to build $N_1(v)$ it is necessary to explore at most twice each edge of $G[N(v)]$ plus one edge for each neighbor.
The set $N_2(v)$ can be constructed in the same way, and then $N_3(v)$ is the remaining neighborhood.
Therefore, the construction of $ N_1(v), N_2(v)$, and $N_3(v)$ is done in $O(d(v))$. Finally, removing the set $N_{2,3}(v)$ and adding $v'$  is also done in $O(d(v))$.

Rule~\ref{rgl: Tot paire} is applied in time $O(d(v)+d(w))$. Indeed, similarly as above, the construction of the three neighborhoods and their transformation is done in $O(d(v)+d(w))$. The construction of the sets $\Dvw, \Dv, \Dw$ is done in $2^{\sqrt{3}}(d(v)+d(w))$ thanks to  the fixed-parameter tractable algorithm used in~\cite{ABF+02} for problems of the form ``\textsc{Planar Dominating Set With Property $\mathcal{P}$}''.

Rule~\ref{rgl: Tot aux} is applied in time $O(d(v)+d(w))$. Indeed, the construction of a simple $vw$-region is done in $O(d(v)+d(w))$. To see this, it suffices to observe that a cyclic order around $v$ and $w$ induces a partial cyclic order on \N{}. A simple region contains exactly a sequence of consecutive vertices in this ordering.

In the worst case, testing whether Rule~\ref{rgl: Tot seul} can be applied on every vertex is done in $ \sum_{v\in V} O(d(v)) = O(n)$. Testing whether Rules~\ref{rgl: Tot paire} and~\ref{rgl: Tot aux} can be applied on every pair of vertices is done in $ \sum_{v,w\in V} O(d(v)+d(w)) = O(n^2)$. Again in the worst case, for each test only one rule can be applied, and this rule removes at least one vertex (see Definition~\ref{def:reduced}).
Hence, all these tests can be repeated at most $n$ times, and
therefore the graph can be reduced in time $O(n^3)$. \end{proof}

\section{Bounding the size of the kernel}
\label{sec:bound}

In this section we bound the size of a plane graph reduced under our rules by a linear function depending on the parameter $k$, hence yielding the desired kernel. To this aim, we assume that the plane graph admits a \tds $D$ (at the end of the proof, it will be of size at most $k$), and we exhibit a $D$-region decomposition $\Re$ such that:
\begin{itemize}
\item[$\bullet$] $\Re$ has at most $3|D|$ regions,
\item[$\bullet$] $\Re$ covers all vertices but $98|D|$ of them, and
\item[$\bullet$] each region $R \in \Re$ contains at most $104$ vertices distinct from its poles.
\end{itemize}

We prove three propositions which respectively correspond to the three properties above. Namely, Proposition~\ref{prop:  taille decompo} is proved using the notions of underlying multigraph and thinness, Proposition~\ref{prop: ext decompo} is proved thanks to Rule~\ref{rgl: Tot paire} and the notion of simple region, and Proposition~\ref{prop: int region} is proved thanks to Rule~\ref{rgl: Tot seul} and also simple regions. These three propositions will be proved in Subsections~\ref{sec:size},~\ref{sec:outside}, and~\ref{sec:inside}, respectively. Finally, in Subsection~\ref{sec:done} we combine the results presented so far and prove Theorem~\ref{th:main}.

\medskip


The following fact bounds the number of vertices in a simple region, and follows directly from the statement of Rule~\ref{rgl: Tot aux} . It will be used in Propositions~\ref{prop: int region} and~\ref{prop: ext decompo}.


\begin{fact}\label{lem: taille reg simple}
Let $G$ be a reduced plane graph. Any simple region $R$ of a region decomposition of $G$ contains at most four vertices distinct from its poles.
\end{fact}

\subsection{Size of the decomposition}
\label{sec:size}

The following proposition bounds the number of regions in a maximal region decomposition. Note that this result does not assume the graph to be reduced.

\begin{proposition} \label{prop: taille decompo}
Let $G$ be a plane graph, let $D$ be a \tds of $G$ of size at least $3$. There exists a maximal $D$-region decomposition $\Re$ of $G$ such that $|\Re|  \leq 3|D| -6$.
\end{proposition}

\newcommand{\BF}{\ensuremath{B(F)}\xspace}
\begin{proof}
Let $\Re$ be a maximal region decomposition obtained with the greedy algorithm presented in~\cite{AFN04} (see also~\cite{GST13}). We first prove that the underlying multigraph $G_\Re = (V_\Re, E_\Re)$ is planar, and next that it is thin. Since a thin multigraph verifies Euler's formula (see~\cite{GST13} for a proof), the claimed bound follows.

We first prove that $G_\Re$ is planar.
Let $\pi$ be the embedding of the plane graph $G$. We consider the mapping $\pi_\Re$ from $G_\Re$ to the plane such that:
\begin{itemize}
\item[$\bullet$] for each vertex $v \in D$, $\pi_\Re (v) = \pi(v)$, and
\item[$\bullet$] for each edge $e \in E_\Re$, $\pi_\Re (e) = \bigcup_{f\in p_e} \pi(f)$, where $p_e$ is an arbitrarily chosen \bound path of the region of $\Re$ corresponding to $e$.
\end{itemize}

 Note that the images of vertices are pairwise disjoint, and that the image of an edge $\{v,w\}$ has $\pi(v), \pi(w)$ as endpoints and does  not contain any other image of vertices in $D$, because $\pi$ is an embedding and any $vw$-region does not contain vertices from $D \setminus \{v,w\}$.

We now proceed to modify $\pi_\Re$ such that images of edges are also pairwise disjoint. For any edge set $F \subseteq E_\Re$, we denote $\pi_\Re(F) = \bigcap_{e\in F} \pi_\Re(e)$. Observe that given $F,F'\subseteq E_\Re$, if $F \subseteq F'$ then  $\pi_\Re(F) \supseteq \pi_\Re(F')$. We modify $\pi_\Re$ according to the following procedure:


\bigskip
\noindent
\While{there exists $F \subseteq E_\Re$ with $|F| \geq 2$ such that $\pi_\Re(F) \not \subseteq \bigcup_{v\in D} \pi_\Re(v)$,}{
\noindent
let $F$ be an edge set, maximal by inclusion such that  $\pi_\Re(F) \not \subseteq \bigcup_{v\in D} \pi_\Re(v)$;\\
\noindent
let \BF be a subset of the plane that, for each point $x$ of a connected part of $\pi_\Re(F) \setminus  \bigcup_{v\in D} \pi_\Re(v)$, contains a ball with center $x$ and radius $\epsilon >0$, with $\epsilon$ small enough so that \BF intersects only a connected part of images of edges from $E$;\\
\ForAll {edges $e \in F$,}{
\noindent
let $p,q$ be the endpoints of $\pi_\Re(e) \cap \BF$;\\
\noindent
let $C_e \in \BF$ be a curve with endpoints $p,q$ such that $(C_e \setminus \{p,q\}) \cap \bigcup_{f \in F} \pi_\Re (f) = \emptyset$;\\
\noindent
set $\pi_\Re(e) := \pi_\Re(e) \setminus \BF \cup C_e$.
 }       }
\bigskip

Let us give some intuition about the sets $B(e)$ and $B(F)$ defined above. First, the set $B(e)$ (with $e \in E_\Re$) is a small subset of the plane around the drawing of $e$. Recall that the drawing of $e$ corresponds to a path in $G$; along this path $B(e)$ can be a ``tube'' of radius $\varepsilon$; around the extremity of this path $B(e)$ looks like a cone that is open at the extremity. The set $B(F)$ generalizes $B(e)$ for the intersection of drawings of edges in $F$. Namely, $B(F)$ is a small subset of the plane around the intersection of drawings of edges in $F$. Intuitively, it looks like a ``tube'' around each point of the drawing, except at extremities of the edges. The motivation for the definition of $B(e)$ is that we want that drawings of the edges of $G_\Re$ to be pairwise disjoint expect at their extremities (that is, at the vertices of $G_\Re$).

Let us now make some observations about the above procedure.
First, when the procedure stops, for each set $F \subseteq E_\Re$, $\pi_\Re(F)$ is empty, or restricted to one or two points, which are images of common endvertices of edges in $F$.
Also, when the procedure is considering a set $F$, the subset \BF exists because edge sets are treated in decreasing order: indeed, since $F$ is inclusion-wise maximal, $\pi_\Re(F)$ cannot intersect an image $\pi_\Re(e)$ for $ e \notin F$ (except for its endpoints).
%
Finally,  when the procedure is considering an edge $e \in F$, the curve $C_e$ exits because the paths $p_f$ for $f \in F$ are confluent, and therefore the endpoints of curves $\pi_\Re(f) \cap \BF$ are nicely ordered around the boundary of \BF, by Definition~\ref{def: confluent}.
When the procedure stops, the images of edges are indeed pairwise disjoint (except for their endpoints), hence $\pi_\Re$ is a embedding of $G_\Re$.
This procedure is illustrated in Fig.~\ref{fig: replonge}.

\begin{figure}[htbp]
\begin{center}
  \includegraphics[scale=1]{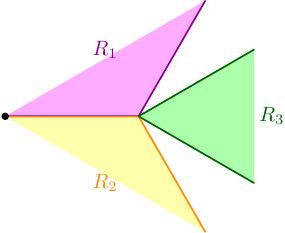}(a)
  \includegraphics[scale=1]{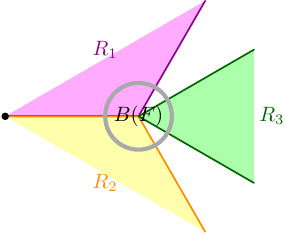}(b)

  \bigskip

  \includegraphics[scale=1]{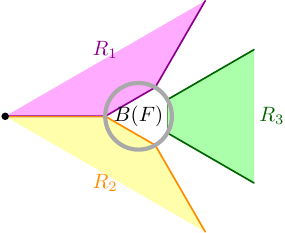}(c)
  \includegraphics[scale=1]{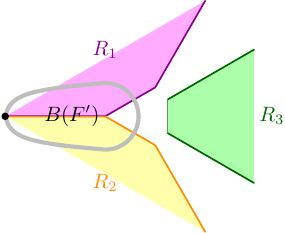}(d)

  \includegraphics[scale=1]{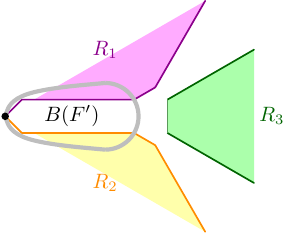}(e)
  \includegraphics[scale=1]{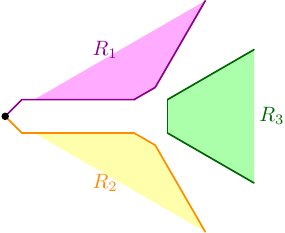}(f)
  \caption{\small{Illustration of the re-embedding procedure of Proposition~\ref{prop: taille decompo}. (a) The edges associated with three regions $R_1,R_2,R_3$ are embedded along the boundary of their regions. (b) The procedure first considers  the set $F$ of the three edges that intersect on a point. (c) The embedding is locally modified inside $B(F)$ around this point. (d) The procedure then considers the set $F'$ of two edges that intersect on a segment. (e) The embedding is locally modified inside $B(F')$ around this segment. (f) The edges do not intersect anymore, hence the embedding is valid.}}
  \label{fig: replonge}
\end{center}
\end{figure}

We now prove that $G_\Re$ is thin. We have to show that each of the two open sets delimited by two edges $e_1,e_2$ with identical endvertices contains a vertex from $D$. In fact, we show that there is a dominating vertex in the two open sets delimited by the corresponding pair of $vw$-regions $R_1,R_2$ with identical poles. 
More precisely, let $O_\Re$ be one the two open sets delimited by $e_1,e_2$. Let $O$ be the open set delimited by $p_{e_1},p_{e_2}$ (the \bound paths chosen above) and that corresponds to $O_\Re$ (that is, the same orientation has been chosen to traverse $e_1,e_2$ and $p_{e_1},p_{e_2}$), minus the closed sets $R_1, R_2$ (see Fig.~\ref{fig: ouvert}). Observe that since the embedding $\pi_\Re$ does not modify the position of vertices with respect to edges of $G_\Re$, $V(O_\Re) = V(O) \cap D$.

\begin{figure}[htbp]
\begin{center}
  \includegraphics[scale=0.95]{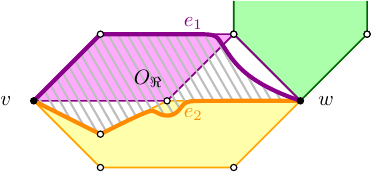}(a)
  \includegraphics[scale=0.95]{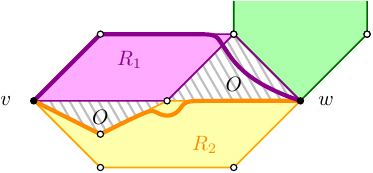}(b)
  \caption{\small{ Illustration of the open sets considered in the proof of Proposition~\ref{prop: taille decompo}. (a) The open set $O_\Re$ delimited by edges $e_1,e_2$. (b) The open set $O$ delimited by the boundaries of $R_1,R_2$.
    }}
  \label{fig: ouvert}
\end{center}
\end{figure}

 Suppose for contradiction that $O$ does not contain any dominating vertex. Vertices in $O$ are necessarily dominated by $v$ or $w$. Therefore, the regions $R_1,R_2$ and regions of $\Re$ included in $O$ can be replaced with the region $R_1 \cup O \cup R_2$,  contradicting the maximality of $\Re$. Since each edge of $G_\Re$ corresponds to a region of $\Re$, the multigraph $G_\Re$ is thin. According to Euler's formula, it follows that $|\Re| \leq 3 |D| - 6$.\end{proof}

Observe that the same argumentation can easily be applied to prove that the underlying multigraph of a simple region decomposition is also thin. In Propositions \ref{prop: ext decompo} and \ref{prop: int region} we will decompose subgraphs into simple regions, and we will need to bound the number of simple regions.

\subsection{Number of vertices outside the decomposition}
\label{sec:outside}

We now proceed to bound the number of vertices which are not covered by the region decomposition, that is, we bound the size of $V\setminus V(\Re)$. Using Rule~\ref{rgl: Tot seul}, we show that vertices outside of the regions can be covered by a small number of simple regions, and then we use Rule~\ref{rgl: Tot aux}.
We would like to note that we will assume the graph  to be reduced, but in fact we only need Rules~\ref{rgl: Tot seul} and~\ref{rgl: Tot aux}.

\begin{proposition} \label{prop: ext decompo}
Let $G$ be a reduced plane graph, let $D\neq \emptyset$ be a \tds of $G$, and let $\Re$ be a maximal $D$-region decomposition of $G$. Then $|V\setminus (V(\Re) \cup D)| \leq 97|D|$.
\end{proposition}

\begin{proof}
We show that $|V\setminus V(\Re)| \leq 32|\Re| + |D|$, and then Proposition~\ref{prop: taille decompo} provides the claimed bound.  Since $D$ dominates $V$, we can consider $V$ as $\bigcup_{v\in D} N(v)$, and thus it suffices to bound the size of $N_1(v),N_2(v), N_3(v)$ minus $V(\Re)$.

We first consider $N_1(v)$, and we show $N_1(v) \subseteq V(\Re)$. Let $u \in N_1(v)$ and assume for contradiction that $u \notin V(\Re)$. We exhibit a path $p=(v,u,\dots,w)$ with $w \in D$ of length at most three that does not cross any region of $\Re$, contradicting the maximality of $\Re$.

By definition of $N_1(v)$, there is a vertex $y \in N(u) \setminus N[v]$. Assume fist that $y \in D$. By hypothesis $u \notin V(\Re)$, and thus $p=(v,u,y)$ does not cross any region of $\Re$.

 Suppose now that $y \notin D$. Let $\mathfrak{P}_y$ be the set of \bound paths of regions in $\Re$ containing $y$. Suppose fist that $\mathfrak{P}_y \neq \emptyset$. Let $\check{u} \in N(y)$ be on a path of $\mathfrak{P}_y $ such that $\check{u}$ is minimal around $y$ from $u$ in the anticlockwise cyclic ordering of the edges containing $y$ induced by the embedding (by hypothesis, $u \neq \check{u}$). Recall that $y$ has to be dominated.

     Assume that $\check{u} \in D$. By construction, the path $ p =( v,u,y,\check{u})$ does not cross any region of $\Re$ (cf. Fig.~\ref{fig: ext decompo}(a)-(b)).
     Assume now that $\check{u}\notin D$. Let $\hat{w} \in N(y)$ on a path $(v',\check{u},y,\hat{w}) \in \mathfrak{P}_y$ such that $\hat{w}$ is maximal around $y$ from $u$, where $\check{u}$ is fixed above and $v'$ is any dominating vertex. By construction, the path $p = (v,u,y,\hat{w})$ does not cross any region of $\Re$ (cf. Fig.~\ref{fig: ext decompo}(c)-(d)).

   Suppose now that $\mathfrak{P}_y = \emptyset$. Since $y$ has to be dominated, there exists $w \in D \cap N(y)$ distinct from $v$. By hypothesis, neither $u$ nor $y$ is in $V(\Re)$, thus the path $p =(v,u,y,w)$ does not cross any region of $\Re$.

In all cases, the path $p$ is a region that could be added to $\Re$, contradicting its maximality.
Therefore, $|\bigcup_{v\in D} N_1(v) \setminus V(\Re)| = 0$.

\begin{figure}[htbp]
\begin{center}
  \includegraphics[scale=0.8]{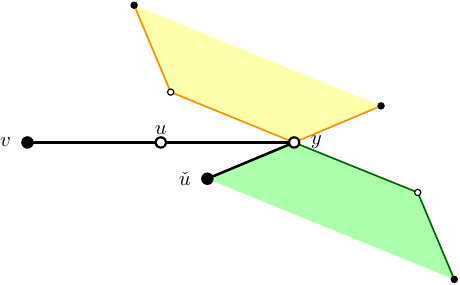}(a)
  \includegraphics[scale=0.8]{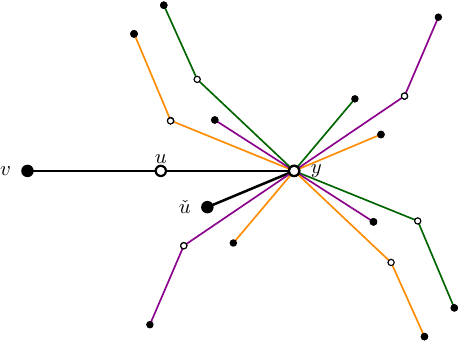}(b)\vspace{.8cm}
   \includegraphics[scale=0.8]{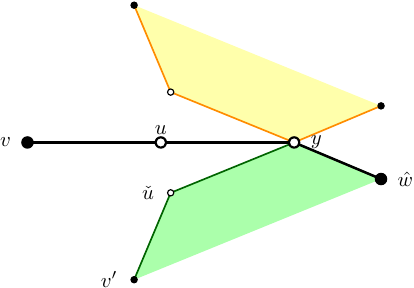}(c)
  \includegraphics[scale=0.8]{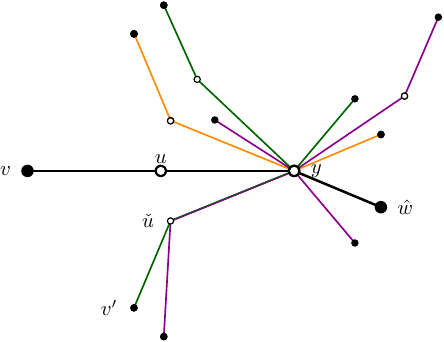}(d)

\vspace{.8cm}

  \caption{\small{ Construction of the path in the proof of Proposition~\ref{prop: ext decompo}, in the case where $y \in V(\Re) \setminus D$.
  (a) A simple configuration where $\check{u}$ is dominating.
  (b) A more involved configuration where regions are degenerate; it is important to choose $\check{u}$ to be minimal.
  (c) A simple configuration where $\check{u}$ is not dominating and it is necessary to look for $\hat{w}$ on the other side.
  (d) A more involved configuration where regions are degenerate; note that edge $\{\check{u},y\}$ belongs to two paths, hence it is important to choose $\hat{w}$ to be maximal.
    }}
  \label{fig: ext decompo}
\end{center}
\end{figure}

\smallskip

We now consider the set $N_2(v)$. We show that $N_2(v) \setminus V(\Re)$ can be always covered with a set $\cal R$ of at most $4 d_{G_\Re}(v)$ simple regions, where $d_{G_\Re}(v)$ is the number of regions in $\Re$ having $v$ as a pole.
Observe that, by definition of $N_2(v)$, any vertex $y \in N_2(v) \setminus V(\Re)$ is, on the one hand, adjacent to $v$ and, on the other hand, adjacent to a vertex in $N_1(v)$. Observe also that such a vertex $u \in N_1(v)$ and the edge $\{u,v\}$ are necessarily on the boundary of a region in $\Re$ having $v$ as a pole. Indeed, we have proved just above that $u$ must be in $V(\Re)$, hence it is on a boundary. Moreover, assume for contradiction that $\{u,v\}$ is not on the boundary of a region. Then, considering $\mathfrak{P}_u$ and the arguments used above, we can find a vertex $w \in D$ such that the path $(v,y,u,w)$ can be added to $\Re$,  contradicting the maximality of $\Re$. Therefore, for each region of $\Re$, there are at most two vertices in $N_1(v)$ adjacent to $N_2(v) \setminus V(\Re)$.

Let $S =  \{ u \in N_1(v) \mid \{v,u \} \in \partial R, R \in \Re\}$, and note that, by definition, $|S| \leq 2  d_{G_\Re}(v)$. Let $\Re_S$ be a maximal $\{v\} \cup S$-simple region decomposition of the graph $G[\{v\} \cup S \cup (N_2(v) \setminus V(\Re))]$.
Such a decomposition exists, covers all vertices in $N_2(v) \setminus V(\Re)$ , and has a thin underlying multigraph $G_{\Re_S}$. To see that $G_{\Re_S}$ is thin, it suffices to follow the argumentation of Proposition~\ref{prop: taille decompo}: if there is no pole between two $vu$-simple regions (with $u \in S$), then these regions could be merged, contradicting the maximality of $\Re_S$. To see that $G_{\Re_S}$ covers $N_2(v) \setminus V(\Re)$, it suffices to follow the argumentation above: if $y \in N_2(v) \setminus V(\Re)$ is not in a simple region of $\Re_S$, then there exists a path $(v,y,u)$ (with $u \in S$) which could be added to $\Re_S$, contradicting the maximality of $\Re_S$.
Moreover, $G_{\Re_S}$ is a multistar centered at $v$. Hence, by Euler's formula, $|E_{\Re_S}| \leq 2 |V_{\Re_S}| -3 \leq 4 d_{G_\Re}(v)$.

Since $\Re_S$ has been defined on the induced subgraph $G[\{v\} \cup S \cup (N_2(v) \setminus V(\Re))]$, in a region $R' \in \Re_S$ there may be images of vertices from $V \setminus (\{v\} \cup S \cup N_2(v))$. These vertices are not adjacent to $N_2(v)$, hence they can be re-embedded outside $\Re_S$. By Lemma~\ref{lem: taille reg simple}, each simple region of $\Re_S$ contains at most four vertices distinct from its poles. Hence, $|\bigcup_{v\in D} N_2(v)\setminus V(\Re)| \leq 4 \sum_{v\in D} 4d_{G_\Re} (v) \leq 32|\Re|$.

Finally, let us consider the set $N_3(v)$. Since $G$ is reduced by Rule~\ref{rgl: Tot seul}, $N_3(v)$ contains at most one vertex. Therefore, $|\bigcup_{v\in D} N_3(v)\setminus V(\Re)| \leq |D|$.

Summing up the upper bounds obtained above for $N_1(v),N_2(v), N_3(v)$  for each vertex of $v \in D$, we obtain that $V \setminus V(\Re) \leq 32|\Re| + |D|$, and the result follows.
\end{proof}


\subsection{Size of the regions}
\label{sec:inside}

The last step toward the proof of Theorem~\ref{th:main} is to bound the number of vertices inside a region. Thanks to Rule~\ref{rgl: Tot paire}, the vertices inside a region either have been removed, or can be dominated by a small number of vertices. In the latter case, we use \N3-\tds to build simple regions. This trick together with Rule~\ref{rgl: Tot aux} yields the desired bound. Again, we would like to note that we will assume the graph  to be reduced, but in fact we only need Rules~\ref{rgl: Tot paire} and~\ref{rgl: Tot aux}.


\begin{proposition} \label{prop: int region}
Let $G$ be a reduced plane graph, let $D$ be a \tds of $G$l and let $v,w \in D$. Any $vw$-region $R$ contains at most 104 vertices distinct from its poles.
\end{proposition}


\begin{proof}
Let $R$ be a $vw$-region. It will become clear from the proof that the worst case is obtained when $\partial R$ is formed with two disjoint (except at their endvertices) paths of length 3, that is, when $\partial R$ contains a maximum number of vertices, which we denote by $v,u_v,u_w,w,u'_w,u'_v$.

We first count the vertices in \N1 inside $R$. These vertices are necessarily on $\partial R$ and, by definition, there are at most six such vertices (including $v,w$).
Hence, there is at most four vertices of \N1 in $R$.

We now count vertices of \N{2,3} in $R$. To this aim, we bound the size of a maximal simple region decomposition covering \N2 and \N3 in $R$. We distinguish five cases: the one where Rule~\ref{rgl: Tot paire} has not been applied on $v,w$, which we call Case~0, and the four cases of this rule. Observe that, in any case, by definition of \N2, every vertex of $\N2 \cap V(R)$  is, on the one hand, adjacent to $v$ or $w$, on the other hand, adjacent to a vertex in $\N1 \cap V(R) \subseteq \{u_v,u_w,u'_v,u'_w \}$.

\begin{enumerate}

\item[0.] $\Dvw \neq \emptyset$ (that is, Rule~\ref{rgl: Tot paire} does not apply). Let $\{d_1,d_2,d_3 \} \in \Dvw$ (it appears that the worst case is obtained with an \N3-\tds of size three). Observe that,  by definition of \N3, every vertex of $\N3$  is, on the one hand, adjacent to $v$ or $w$ and, on the other hand, adjacent to a vertex in $\{d_1,d_2,d_3 \}$.  Let $S = (\{ u_v,u_w,u'_v,u'_w \} \cup \{ d_1,d_2,d_3 \}) \cap V(R)$, 
    and note that $|S| \leq 7$.
 Consider a maximal $\{v,w\} \cup S$-simple region decomposition $\cal R$ of $G[V(R)]$.
Such a decomposition exists, covers the set $\N{2,3} \cap V(R)$, and has a thin underlying multigraph $G_{\cal R}$ (we skip this proof, which is similar to the construction of $\cal R$ in Proposition~\ref{prop: ext decompo}).
 Moreover, according to the previous observations about the neighborhoods of \N2 and \N3, we can build $\cal R$ such that $G_{\cal R}$ is bipartite, with bipartition $\{v,w\} \cup S$.
By Euler's formula, $|E_{\cal R}| \leq 2| \{v,w\} \cup S | -4 = 14$. Observe that simple regions in $G[V(R)]$ are also simple regions in $G$.
By Lemma~\ref{lem: taille reg simple}, a simple region has at most four vertices distinct from its poles. As each edge of $G_{\cal R}$ corresponds to a region of $\cal R$, in this case $R$ contains at most $ 4|{\cal R}| + |S| \leq 4 \cdot 14 + 7 = 63$
vertices distinct from $v,w$. 

\smallskip
\item[1.]
$\Dv = \emptyset$ and $ \Dw = \emptyset$ (that is, Case~\ref{cas: tot et} of Rule~\ref{rgl: Tot paire} applies). Hence, \N2 has been removed and \N3 has been replaced with at most three vertices ($v',w'$, and possibly $y$).
Therefore, in this case $R$ contains at most $4 + 3 = 7$ vertices distinct from $v,w$.

\smallskip
\item[2.]
$\Dv \neq \emptyset$ and $ \Dw \neq \emptyset$ (that is, Case~\ref{cas: tot ou} of Rule~\ref{rgl: Tot paire} applies). Let $\{v, d_v, d'_v \} \in \Dv$ and $\{w, d_w, d'_w \} \in \Dw$ (it will appear that the worst case is obtained with two \N3-\tds of size three).
Observe that every vertex in $\N3 \setminus N(v)\cap N(w)$ is, on the one side, adjacent to $v$ or $w$ and, on the other side, adjacent to a vertex in $\{v, d_v, d'_v,  w, d_w, d'_w \}$. Let $S = (\{ u_v,u_w,u'_v,u'_w \} \cup \{ d_v, d'_v, d_w, d'_w \}) \cap V(R)$, and note that $|S| \leq 8$. Consider  a maximal $\{v,w\} \cup S$-simple region decomposition $\cal R$ of $G[V(R)]$.
Such a decomposition exists, covers all $\N{2,3} \cap V(R)$, and has a thin underlying multigraph $G_{\cal R}$. By Euler's formula, $|E_{\cal R}| \leq 3| \{v,w\} \cup S | -6 = 24$.
Observe that simple regions in $G[V(R)]$ are also simple regions in $G$ (as in the previous case). Therefore, using  Lemma~\ref{lem: taille reg simple}, in this case $R$ contains at most $ 4|{\cal R}| + |S| \leq 4 \cdot 24 + 8 = 104$ vertices distinct from $v,w$.

\smallskip
\item[3.]
$\Dv \neq \emptyset$ and $ \Dw = \emptyset$ (that is, Case~\ref{cas: tot v} of Rule~\ref{rgl: Tot paire} applies). We consider separately the set $\N3 \cap N(v)$ and all the other vertices. We first bound the sizes of $\N3 \setminus N(v)$ and \N2 in $R$.  Let $ \{v,d,d'\} \in \Dv$ (again, it will appear that the worst case is obtained with an \N3-\tds of size three). Note that every vertex in $\N3 \setminus N(v)$ is, on the one side, adjacent to $w$ and, on the other side, adjacent to a vertex in $\{d,d' \}$. Let $S = (\{ u_v,u_w,u'_v,u'_w \} \cup \{ d,d'\}) \cap V(R)$, and note that $|S| \leq 6$. Consider  a maximal $\{v,w\} \cup S$-simple region decomposition $\cal R$ of $G[V(R)]$. Such a decomposition exists, covers all $\N{2,3} \cap V(R)$, and has a thin underlying multigraph $G_{\cal R}$.  Moreover, according to the observations about the neighborhoods of \N2 and \N3, we can build $\cal R$ such that $G_{\cal R}$ is bipartite, with bipartition $\{v,w\} \cup S$). By Euler's formula, $|E_{\cal R}| \leq 2| \{v,w\} \cup S | -4 = 12$.
Again, simple regions in $G[V(R)]$ are also simple regions in $G$.

We now bound the size of $ N(v) \cap \N{2,3}$. We show, in two stages, that there are at most 10+1 vertices in $N(v) \cap \N{2,3}$. Recall that, except for the added vertex $v' \in N_3(v)$, vertices in $ N(v) \cap \N{2,3}$ are in $\bigcup \Dv$, and thus in $N_1(v)$ as well. We bound the number of vertices with neighborhoods intersecting $N[d]$, and then we use same argumentation for vertices with neighborhoods intersecting $N[d']$. Assume for contradiction that there are at least five such vertices and let us call them $u_1,u_2,u_3,u_4,u_5$. Remark that the vertices in $N(d) \cap N(w)$ are separated from $v$ and from $u_i$ for $ i \in [1,5]$ by two $dw$-paths. Therefore, at least three of the $u_i$'s are adjacent to a common vertex in $N(d)$. Assume without loss of generality that these vertices are $u_1,u_2,u_3$. Since $G$ is reduced by Case~\ref{cas: tot v} of Rule~\ref{rgl: Tot paire}, the sets $N(u_i) \setminus N(v)$ for $i \in [1,5]$ are pairwise incomparable, and thus $u_1,u_2,u_3$ are adjacent to vertices in $N(d') \cap N(v)$ or in $\{ w,u_w,u'_w\}$. Contracting, on the one side, edges of the form $\{d,y\}$ with $y \in N(d) \cap N(v)$ on $d$ and, on the other side, edges of the form $\{w,z\}$ with $z \in N(w) \setminus N(d)$ on $w$, we obtain a $K_{3,3}$-minor in $G$,  contradicting its planarity (cf. Fig.~\ref{fig: Tot mineur} for an illustration). Hence, $|\bigcup \Dv| \leq 10$. Therefore, using Lemma~\ref{lem: taille reg simple}, in this case $R$ contains at most $ 4|{\cal R}| + |S| +10+1 \leq 4 \cdot 12 + 6 +10+1 = 65$ vertices distinct from $v,w,v',w'$.

\begin{figure}[htbp]
 \begin{center}
 \includegraphics[scale=0.85]{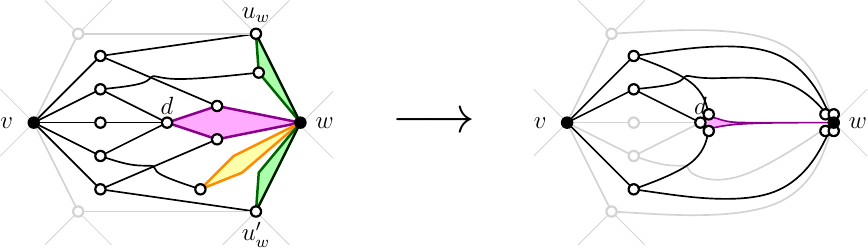}
 \caption{\small{Exhibition of a $K_{3,3}$-minor in the proof of Proposition~\ref{prop: int region}, contradicting planarity. }}
 \label{fig: Tot mineur}
 \end{center}
\end{figure}

\smallskip
\item[4.]
$\Dv = \emptyset$ and $ \Dw \neq \emptyset$ (that is, Case~\ref{cas: tot w} of Rule~\ref{rgl: Tot paire} applies). Symmetrically to Case~\ref{cas: tot v}.

\end{enumerate}

Thus, $R$ contains at most $\max \{63,7,104,65\} = 104$ vertices distinct from $v,w$. 
\end{proof}

\begin{figure}[htbp]
 \begin{center}
 \includegraphics[scale=0.85]{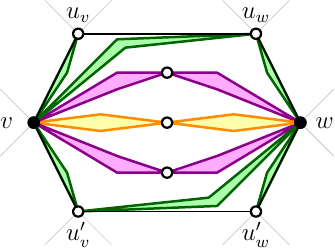}(a)
 \includegraphics[scale=0.85]{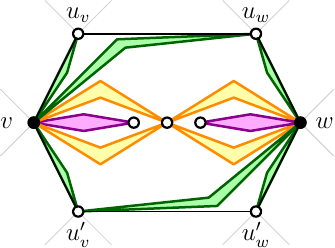}(b)

 \vspace{.5cm}

 \includegraphics[scale=0.85]{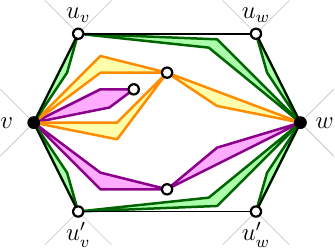}(c)
 \includegraphics[scale=0.85]{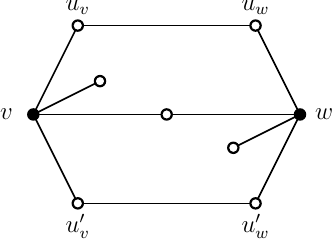}(d)

\vspace{.5cm}

 \includegraphics[scale=0.85]{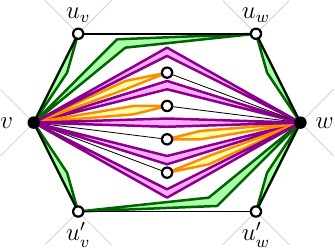}(e)
 \includegraphics[scale=0.85]{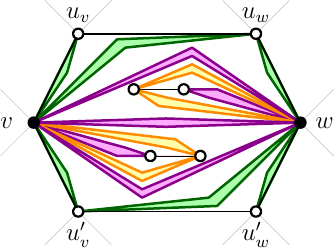}(f)

\vspace{.5cm}

 \includegraphics[scale=0.85]{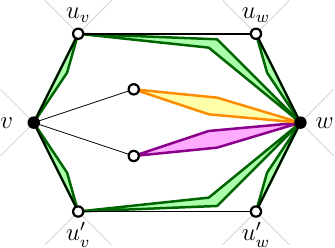}(g)
 \includegraphics[scale=0.85]{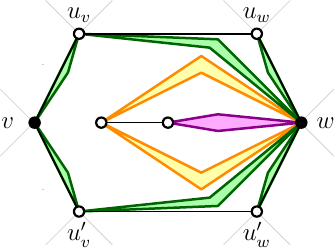}(h)

\vspace{.5cm}

 \caption{\small{Some illustrations of $\mathcal{R}$ covering $N_2(v,w)$ and $N_3(v,w)$ in Proposition~\ref{prop: int region}.
  (a)-(b)-(c) Cases in which Rule~\ref{rgl: Tot paire} does not apply, with poles in \Dvw.
  (d)     Case~\ref{cas: tot et}, where all vertices are replaced with gadgets.
  (e)-(f)   Case~\ref{cas: tot ou}, with poles in \Dv and \Dw.
  (g)-(h)   Case~\ref{cas: tot v},  with poles in \Dv.
  Note that the number of simple regions in the figure is smaller than the bound in the proof of Proposition~\ref{prop: int region}; however we think that the depicted configurations are the worst possible cases.
  }}
 \label{fig: Tot nb interne}
 \end{center}
\end{figure}

Let us make some remarks about the proof of Proposition~\ref{prop: int region} toward eventual improvements to obtain a better bound. Firstly, in order to obtain the worst possible cases, we have assumed that the considered region $R$ contains the three vertices of an \N3-\tds. In fact, these vertices may be spread over all $vw$-regions.
Secondly, Case~\ref{cas: tot ou}, which gives the largest bound, forbids the existence of vertices $v',w'$, which have already be counted in Proposition~\ref{prop: ext decompo}.
Thirdly, we think that the general Euler's formula for bipartite multigraphs provides a rough bound, as the underlying multigraphs are very constrained by the topology of a region (cf. Fig.~\ref{fig: Tot nb interne}).
Fourthly, depending on the configuration, it is possible to provide smaller bounds on the sizes of simple regions \cite{AFN04,GST13}.
Summarizing, we are convinced that with a finer analysis, it is possible to provide a better bound in Proposition~\ref{prop: int region}. Nevertheless, as the proof is already quite technical, we decided to keep the current formulation.

\subsection{Putting the pieces together}
\label{sec:done}

We are finally ready to provide the proof of Theorem~\ref{th:main}.


\bigskip

\begin{proofFINAL}
We show that the exhaustive application of the reduction rules, in any order, defines a polynomial-time algorithm that, for any planar instance $(G,k)$, returns an equivalent planar instance $(G',k)$ such that either $G$ and $G'$ are \textsc{No}-instances, or $|V(G')| \leq 410 \cdot k$. Let $G$ be the input plane graph, and let $G'$ be the graph reduced by Rules~\ref{rgl: Tot seul},~\ref{rgl: Tot paire}, and~\ref{rgl: Tot aux}. According to Lemmas~\ref{lem: Tot seul correct},~\ref{lem: Tot paire correct}, and~\ref{lem: Tot aux correct}, $G'$  has a \tds of size at most $k$ if and only if $G$ has one. By Lemma~\ref{lem: cplx}, the reduction rules can be applied in polynomial time.
Let us now analyze the size of the kernel. Assume first that $G'$ admits a solution $D'$. If $|D'| = 0$, then $G'$ is empty. It cannot be that $|D'| \neq 1$, because $D'$ needs to be a \tds, and if
 $|D'| = 2$, then $G'$ can be covered with only one region, and the claimed bound clearly holds. Otherwise, $|D'| \geq 3$, and we can apply Propositions~\ref{prop: taille decompo},~\ref{prop: ext decompo},~\ref{prop: int region}, and deduce that the size of $G'$ is bounded by $104(3k -6)+ 97k + k < 410 \cdot k$. If $|V(G')| > 410 \cdot k$, then it follows that $G$ is a \textsc{No}-instance, and $G'$ can be  replaced with any constant-sized \textsc{No}-instance. Hence the reduced instance is indeed a kernel of the claimed size.
\end{proofFINAL}

\section{Further research}
\label{sec:concl}

As we showed in Section~\ref{sec:prelim}, \textsc{Total Dominating Set} satisfies the general conditions of the meta-theorem of Bodlaender et al.~\cite{BFL+09}, and therefore we know that there {\sl exists} a linear kernel on graphs of bounded genus. Finding an {\sl explicit} linear kernel on this class of graphs seems a feasible but involved generalization of our results. In this direction, it would be interesting to see whether \textsc{Total Dominating Set} fits into the recent framework of Garnero et al.~\cite{GPST14,GarneroPST16} for obtaining explicit and constructive meta-kernelization results on sparse graph classes.

It is worth mentioning that neither \textsc{(Connected) Dominating Set} nor \textsc{Total Dominating Set}  satisfy the general conditions of the meta-theorems for $H$-minor-free graphs~\cite{FLST10} and $H$-topological-minor-free graphs~\cite{KLP+12}. But Fomin et al. proved  that \textsc{(Connected) Dominating Set} has a linear kernel on $H$-minor-free graphs~\cite{FLST12} and, more generally, on $H$-topological-minor-free graphs~\cite{FLST12b}. Is it also the case of \textsc{Total Dominating Set}?

As mentioned in the introduction, it was not the objective of this article to optimize the running time of our kernelization algorithm. A possible avenue for further research is to apply the recent methods~\cite{Hag11,BHK+11} for obtaining linear kernels for \textsc{Planar Dominating Set} in linear time. 


It would be interesting to use similar techniques to obtain explicit linear kernels on planar graphs for other domination problems. For instance, \textsc{Independent Dominating Set} admits a kernel of size $O(jk^i)$ on graphs which exclude $K_{i,j}$ as a (not necessarily induced) subgraph~\cite{PRS09}; this implies, in particular, a cubic kernel on planar graphs. It is known that \textsc{Independent Dominating Set} has a polynomial kernel on graphs of bounded genus~\cite{BFL+09}, even if the problem does not have FII. The existence of a linear kernel on planar graphs remains open. There are other variants of domination problems for which the existence of polynomial kernels on sparse graphs has not been studied yet, like \textsc{Acyclic Dominating Set}~\cite{HHR00,BGM+12} or $\alpha$-\textsc{Total Dominating Set}~\cite{HeRa12}.

Concerning approximation, \modifOK{}{\totdom admits a  $(\ln(\Delta-0.5)+1.5)$-approximation, where $\Delta$ is the maximum degree of the input graph~\cite{Zhu09}. Moreover,} it may be possible that \modifOK{\textsc{Total Dominating Set}}{it} admits a PTAS on planar graphs, as it is the case of many other graph optimization problems~\cite{Bak94}. \modif{ See also~\cite{Zhu09} for other results about the computational complexity of \textsc{Total Dominating Set}.}{}\\



\acknowledgements
\label{sec:ack}
We would like to thank Pascal Ochem for the proof of Theorem~\ref{th:NPcomplete}. \modifOK{}{We also thank the anonymous reviewers for carefully reading the paper and for their very helpful suggestions, in particular concerning Rule~\ref{rgl: Tot aux}.}

\bibliographystyle{alpha}
\bibliography{linearKernels}
\label{sec:biblio}

\end{document}